\documentclass[12pt,aps,prx,reprint,floatfix,longbibliography]{revtex4-2}

\usepackage{graphicx}
\usepackage{amsthm}
\usepackage{amsmath}
\usepackage{amssymb}
\usepackage{braket}
\usepackage{hyperref}
\usepackage[caption=false]{subfig}
\usepackage[export]{adjustbox}
\newtheorem{thm}{Theorem}
\newtheorem*{thm*}{Theorem}
\newtheorem{cor}{Corollary}[thm]
\newtheorem{lemma}[thm]{Lemma}
\newtheorem{definition}{Definition}

\usepackage{xcolor}

\begin{document}
    \title{Universal Fault Tolerance with Non-Transversal Clifford Gates}
    \author{Benjamin Anker}
    \affiliation{Department of Electrical \& Computer Engineering and Center for Quantum Information and Control, University of New Mexico, Albuquerque, NM 87131, USA}
    \author{Milad Marvian}
    \affiliation{Department of Electrical \& Computer Engineering and Center for Quantum Information and Control, University of New Mexico, Albuquerque, NM 87131, USA}
    \begin{abstract}
        We propose a scheme for the fault-tolerant implementation of arbitrary Clifford circuits. To achieve this, we extend previous work on flag gadgets for syndrome extraction to a general framework that flags any Clifford circuit. This framework opens new pathways toward universal fault tolerance by allowing transversal implementation of $T$ gates alongside fault-tolerant realization of selected non-transversal Clifford gates using flags. The construction we present allows a Clifford circuit consisting of $n$ two-qubit gates and $O(n)$ single-qubit gates acting upon physical qubits in a code of distance $d$ to be made fault tolerant to distance $d$ using $O(d^2 \log(nd^2\log n))$ ancilla qubits and $O(nd^2 \log(nd^2 \log n))$ extra CNOTs. Beyond asymptotic analysis, we demonstrate our construction by implementing the non-transversal logical Hadamard gate for the [[15,1,3]] code, which has transversal T,  and compare to alternative approaches for universality using this code. We also apply our construction to magic-state preparation, general state preparation using Clifford circuits, and data-syndrome codes.

    \end{abstract}
    \maketitle
    \section{Introduction}
    Although quantum computing holds significant promise for many practically relevant purposes, actualizing these promises requires addressing the problem of noise.
    Current hardware is noisy enough that attempting to perform a logical computation without error correction will produce a result which is vanishingly, in the size of the computation, unlikely to be correct. 
    Although error mitigation techniques can help to ensure the correct answer is obtained, they do so with significant overhead~\cite{takagi2022mitigation}.

    Error correction aims to ensure resilience to physical errors by encoding relatively few logical qubits in relatively many physical qubits. 
    The redundancy inherent in this approach allows physical errors to be corrected before they become logical errors.
    Recent progress has convincingly shown that current fabrication processes can in principle extend the lifetime of a logical qubit effectively arbitrarily, provided increasingly large chips can be built~\cite{acharya2022suppressingquantumerrorsscaling}.
    What has not been demonstrated, however, is a universal set of fault-tolerant logical operations.
    This gap is broadly due to the large overhead associated with implementing such a gate set.
    
    If logical operators are designed with only the ideal, error-free setting in mind, an error at the wrong time can become a logical failure -- even though a different error of equal severity might have been harmlessly corrected.
    For this reason we need to ensure our logical operations are fault tolerant. 
    The most straightforward method for ensuring an operation is fault tolerant is to make it transversal, ensuring that no qubits within the same codeblock are coupled.
    Unfortunately, it is impossible to implement a universal gate set entirely transversally for any code with nontrivial distance~\cite{eastin2009transversal}.
    
    One conventional technique to achieve fault tolerance without relying solely on transversal gates is to implement all Clifford gates transversally and supplement with magic states for non-Clifford gates~\cite{bravyi2005magic}; another technique is code switching~\cite{nielsen2012quantum, paetznick2013switch, daguerre2025switch, butt2024switch}.
    Although more efficient approaches to magic state production than distillation have recently been proposed~\cite{gidney2024cultivation, sahay2025cultivation}, their overhead remains significant, and it is not clear how much they can be generalized beyond the surface code.
    
    These challenges are not confined to universality alone; even basic tasks such as syndrome extraction (or just measuring a Pauli operator) face related difficulties.
    A low-weight error can propagate through couplers used to attach an ancilla system to the data, leading to a high-weight error on the data if sufficient care to the design of the coupling circuit is not used.
    This application originally motivated the invention of flag qubits~\cite{reichardt2018flag}.
    The main insight in this work was that, due to the fact that the propagation of Pauli operators/errors through Clifford circuits can be efficiently simulated classically~\cite{gottesman1998simulation}, it is not always necessary to stop errors from propagating immediately.
     Rather, for syndrome extraction, it is enough to `flag' particularly harmful errors in such a way that their effects can be corrected in the near future. 
     This paradigm has been extended in many~\cite{prabhu2023cat, forlivesi2025flag, chamberland2020triangular, tansuwannont2020flag, debroy2020cpc} different directions, including in previous work by the authors of this work.

    The insight we point out above, the use of the Gottesman-Knill theorem, has not yet been exploited fully.
    In theory, one could apply this result to any Clifford circuit, adding structure in the form of flag qubits to ensure harmful errors in the middle of the circuit are noticed.
    Some work in this direction has been done, in particular for (sub-)circuit verification~\cite{debroy2020cpc, delfosse2025clifford, Tham2025optimizedclifford, martiel2025lowoverheaderrordetectionspacetime}.
    However, constructions up to this point have focused upon error mitigation.
    Although the techniques and results are useful, they are not \emph{fault-tolerant} in the arbitrarily-scalable sense.
    
    In this work, we introduce improved gadgets that can be understood as fault-tolerantly measuring a generating set of space-time stabilizers of the circuit~\cite{delfosse2023spacetime, delfosse2025clifford}.
    We show constructively that any Clifford circuit on $n$ two-qubit gates acting within a distance-$d$ code can be made fault tolerant to distance $d$ using only $O(d^2 \log(nd^2 \log n))$ ancilla qubits and $O(nd^2 \log(nd^2 \log n))$ extra CNOTs.
    
    The mere fact that such a construction exists is already nontrivial, since in general one cannot assume a fault-tolerant equivalent circuit exists without code switching~\cite{nielsen2012quantum}.
    Perhaps the most practically interesting part of this result is the alternative avenues this opens for universal fault-tolerant quantum computation and fault-tolerant operations on encoded quantum data.
    That is, if we sacrifice transversality for a Clifford gate instead of the $T$ gate, as in the 3D surface code~\cite{vasmer20193d} or the tetrahedral code~\cite{knill1996reedmuller, koutsioumpas2022smallestT}, then non-Clifford gates can be implemented cheaply transversally, while Clifford circuits for the non-transversal gate can be made fault tolerant with low overhead using flag gadgets.

    In addition to universal computation, our framework applies broadly to other Clifford subroutines.
    We highlight in particular stabilizer state preparation, including for magic state preparation, and applications to data-syndrome codes~\cite{fujiwara2014datasyndrome, ashikmhin2014robus, brown2024datasyndromebch}, where our construction makes fault tolerant certain reductions to the space/measurement overhead of repeated syndrome extraction. 

    We briefly summarize the numerical results we find in a few concrete examples.
    We find that for the 3D color code we can reduce the number of measurements of non-resettable qubits by a factor of $4$, if we are willing to increase the two-qubit gate count by a factor of $20$.
    A comparison of flagged state preparation applied to the magic state distillation subroutine shows that we can reduce the qubit count by a factor of $10$ if we are willing to increase the two-qubit gate count by a factor of $5$. When applying flag qubits to make a data-syndrome code version of the $[[5, 1, 3]]$ code fault-tolerant (beyond simply making the syndrome measurement fault tolerant), we show that we can reduce the qubit count by a factor of $4$ and reduce the CNOT count by a factor of $1.6$.

    Beyond resource counts, however, the conceptual message is that flag gadgets allow the fault tolerance of Clifford circuits to be reframed as a problem of error correction, which can be tackled with well-developed coding-theoretic tools.

    In Section~\ref{sec:background} we give a more thorough introduction to error correction in general, as well as the specific challenges of (quantum) fault tolerance and universality.
    In Section~\ref{sec:general_flags} we review the previous work on flag gadgets and introduce the main concepts and tools we will be using.
    In Section~\ref{sec:flagged_cliffords} we generalize to Clifford circuits and give our main result.
    In Section~\ref{sec:logical_H} we apply our construction to the implementation of logical $H$ for the $[[15, 1, 3]]$ code, which has transversal $T$, as compared to other methods for universality using this code.
    In Section~\ref{sec:msp} we apply our construction to magic state preparation, providing a path to universality resembling a hybrid between the approach in the previous section and conventional approaches, as well as to general state preparation as one commonly used Clifford circuit.
    In Section~\ref{sec:data_syn}, we apply our construction to make a certain method for making syndrome extraction robust to measurement error also robust to mid-circuit errors.
    Finally, in Section~\ref{sec:conclusion} we give our perspective on the results we have achieved and the further questions we have not addressed.
    \section{Background}\label{sec:background}
    Even optimistic estimates suggest that future quantum computers will be roughly $10^{12}$ times noisier than standard classical hardware.
    As such, error correction is much more important for quantum computing than classical computing. 
    In general error correction aims to encode some information redundantly so that there is a way to cross-check for inconsistencies or errors.
    Upon noticing these errors, they can be corrected.
    \subsection{Linear Block Code and Stabilizer Codes}
    One of the most well studied types of classical error-correcting codes is the binary linear block code.
    In this framework, information is encoded as a binary vector $v$ and transformed into a codeword $Gv$ via a generator matrix $G$; the set of all such codewords, the codespace, is simply the image of $G$. 
    Alternatively, we can describe the code using the parity-check matrix $H$, which is related to $G$ through the condition $HG = 0$.
    This allows us to define the codespace equivalently as the nullspace of $H$.
    Crucially, for any codeword $c$ and any error $e$, the syndrome $H(c + e) = He$ depends only on the error $e$, since $Hc = 0$ by construction.

    By considering the parity-check matrix picture, we can generalize to stabilizer codes.
    In this picture, we effectively promote each parity-check to an operator version; instead of defining the codespace as the nullspace of $H$, i.e. the space of vectors $v$ such that $\langle h, v \rangle = 0$ for each row $h$, we define the the codespace as the space of states $\ket \psi$ such that $g_i\ket \psi = \ket \psi$ for all \emph{stabilizer generators} $g_i$, which together generate the \emph{stabilizer group} $G = \langle g_i \rangle$.
    The stabilizer generators take the place of parity checks in this picture.
    This analogy becomes especially transparent for CSS codes~\cite{calderbank1996css, steane1996css}, where each stabilizer generator is either $X$-type or $Z$-type. 
    In this case, a $Z$-type stabilizer $g$ detects $X$ errors in exactly the same way that a parity check detects flipped bits in the classical code. 
    Whether the code is CSS or not, just like in the classical case the syndrome -- defined as the set of eigenvalues obtained by measuring each stabilizer generator on a potentially faulty state $\ket{\psi}$ -- depends only on the errors affecting $\ket{\psi}$, not on the original code state itself.
    Since arbitrary errors can be decomposed in the Pauli basis and $XZ = Y$ up to a phase, for the stabilizer that have no $Y$ term (i.e. CSS codes ) 
    it is enough to correct $X$ and $Z$ errors.
    
    The stabilizer group defines the properties of the code, the most important of which are the number of logical qubits encoded, and the distance of the code. 
    A code on $n$ qubits encoding $k$ qubits which can distinguish any error of weight up to $d - 1$ from the no-error case is summarized by calling it an $[[n, k, d]]$ code. 
    The stabilizer group also implicitly defines the logical operators.
    The logical operators of a code defined by $G$ are just the operators which normalize $G$ -- the nontrivial operators are those which normalize $G$ but are not contained in $G$.
    The minimum weight logical operator is necessarily of weight equal to the distance.
    We can also define a larger code just by putting two smaller codes side by side -- each of the smaller codes is called a codeblock.
    \subsection{Fault Tolerance}
    The key distinction between quantum and classical error correction is that in quantum error correction, we never decode.
    In the classical case, it is reasonable to imagine perfectly computed information $v$ being sent through a noisy channel in encoded form as $Gv$ for protection.
    After transmission, errors can be corrected using the parity checks defined by $H$, followed by decoding to recover and use $v$.
    In the quantum case, however, noise affects not only the communication channel but also the preparation and use of $v$ itself.
    As a result, we require more than just quantum memory—which preserves information through noisy channels over time—we also need fault-tolerant computation, where information remains redundantly encoded throughout the entire computational process.

    The fact that we must implement logical operations on a \emph{noisy} state means that, if we are careless, the logical operations themselves can turn a small error of weight less than $d$, which would be correctable if it occured at a different time, into a logical error of weight greater than or equal to $d$.
    That is, any logical operator which couples multiple physical qubits in the same codeblock runs the risk that an error on one of the coupled qubits will spread to one or more of the other qubits. 

    Avoiding this outcome is the domain of \emph{fault tolerance}. The definition we take here is the following:
    \begin{definition}[Strongly fault-tolerant logical gates on one codeblock]\label{def:FT}
        Suppose the error-free gate $L$ acts upon an error-free state $\ket \psi$ to produce $\ket \phi$. We consider an implementation of $L$ containing $s$ Pauli errors, which we call $\tilde L$, that acts upon a faulty state $\ket {\tilde\psi}$. If $\ket{\tilde \psi}$ is within weight-$r$ Pauli correction to $\ket \psi$ and 
        \begin{enumerate}
            \item if $r + s \leq \frac{d - 1}2$ then $\tilde L \ket{\tilde \psi}$ differs from $\ket \phi$ by a Pauli correction of weight at most $r + s$
            \item if $r + s > \frac{d - 1}2$ then $\tilde L \ket{\tilde \psi}$ differs from \emph{some} codestate by a Pauli correction of weight at most $r + s$
        \end{enumerate}
        then $\tilde L$ is a implementation of $L$ that is \emph{fault-tolerant} to distance $d$.
    \end{definition}

    The definition is similar when we consider multi-qubit gates between codeblock, although we allow the resulting codestate to differ by $r + s$ in each codeblock.
    
    Some definitions (e.g.~\cite{campbell2019singleshot}) relax the requirements, only demanding that the number of resulting errors be proportional to the number of physical errors.
    Regardless of the definition, the core goal of fault tolerance is to ensure that the code experiences noise consistent with the physical error rate, even during computation on encoded states.
    Should fault tolerance be achieved, arbitrarily long quantum computation with arbitrarily low error rate can also be achieved with polylogarithmic overhead~\cite{knill1998threshold,aharonov2008threshold, campbell2017roads}.
    
    One of the most straightforward ways to ensure an implementation of a logical gate is fault-tolerant is to require transversality, i.e. to ensure that the circuit implementing the logical gate $L$ avoids coupling qubits within the same codeblock.
    This guarantees that any error of weight-$w$ introduced during the circuit will have weight no greater than $w$ in any given codeblock after the gate is applied.

    Before proceeding, let us consider which set of gates we need to implement: our goal is to achieve universality, meaning the ability to perform any quantum operation on the encoded logical qubits.
    A common universal gate set is Clifford + $T$, consisting of the Clifford gates -- which map Pauli operators to Pauli operators under conjugation -- along with the $T$ gate, a $\pi/4$ rotation around the $Z$ axis.
    Unfortunately, transversal gates alone are not sufficient to realize a universal gate set~\cite{eastin2009transversal}.
    Therefore, if we aim for universality, at least one gate in the Clifford + $T$ set must be implemented by other means.
    Typically, this gate is the $T$ gate: Clifford gates can then be implemented ``easily'' via transversal operations, while the $T$ gate requires more resource-intensive techniques such as magic state distillation~\cite{bravyi2005magic} or code switching~\cite{paetznick2013switch}.
    In what follows, we re-examine this strategy. 
    We will use flag gadgets heavily throughout the remainder of the work, so we first summarize the core intuition below.
    \subsection{Flag Gadgets}\label{sec:general_flags}
    Briefly put, flag gadgets aim to detect errors that could propagate to high weight errors and correct them, rather than ensuring that no such errors are possible. This idea was first introduced by Reichardt and Chao~\cite{reichardt2018flag}.
    In the original work, and most other works up to this point, the idea has been limited to apply only to syndrome extraction, and has not been shown to be fault tolerant for any other circuit.
    Although in the remainder of this work we will show that their applicability is much broader, syndrome extraction is a good example to demonstrate the fundamental ideas and the utility of flag gadgets.

    In syndrome extraction, we aim to measure some generating set of stabilizers for the stabilizer group we outlined in Section~\ref{sec:background}.
    To illustrate flag qubits, it is enough to consider measuring a single such operator composed of the tensor product of $X$ operators.
    A straightforward implementation of the measurement of $X^{\otimes W}$, where $X^{\otimes W} := \bigotimes_{w \in W} X_w$, is just to prepare an ancilla $a$ in the $\ket +$ state and perform $CNOT_{a, w}$ for each $w \in W$ before measuring the ancilla in the $X$ basis, illustrated in Figure~\ref{fig:syndrome_ex}.
    \begin{figure}
        \centering
        \subfloat[An example of measuring $X^{\otimes\{1, 2, 3, 4\}}$ using a single ancilla qubit. We see that any odd-weight $Z$ error (single qubit $Z$ errors marked in blue) propagates up to the ancilla qubit and causes the measurement to give $-1$, as illustrated by the magenta outline.\label{fig:syndrome_ex}]{
            \includegraphics[width=0.8\columnwidth]{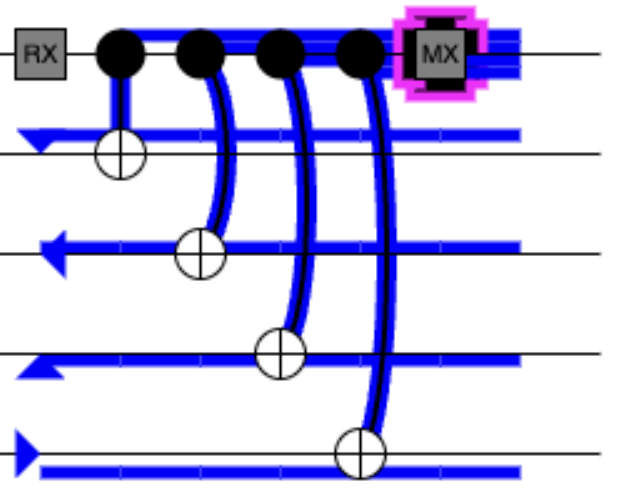}
        }\\[1em]
        \subfloat[One example of non-fault-tolerant propagation of an error from the ancilla qubit. The propagation of the $X$ error that occurs directly before the second CNOT is marked in red.\label{fig:bad_prop}]{
            \includegraphics[width=0.8\columnwidth]{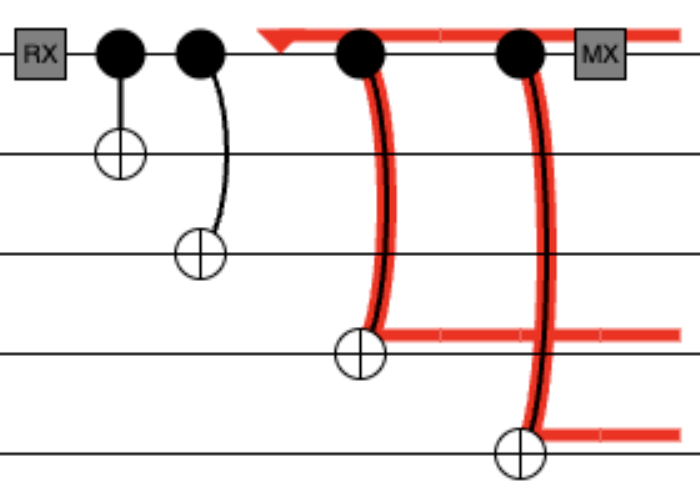}
        }\\[1em]
        \subfloat[\label{fig:syn_flagged}Now the error which propagates non-fault-tolerantly is `flagged' by the additional ancilla, and its propagation can be corrected.]{
        \includegraphics[width=0.8\columnwidth]{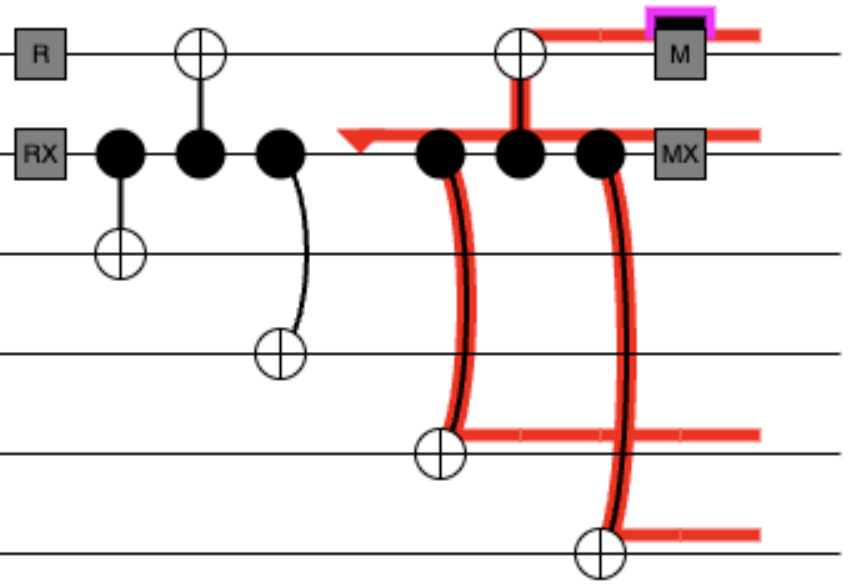}}
        \caption{An illustration of how a straightforward implementation of the measurement of a simple operator can fail to be fault-tolerant, and how fault-tolerance can be restored using a flag qubit. Figures produced using Crumble~\cite{gidney2021stim}.}
    \end{figure}

    Of course, considering fault tolerance makes it obvious why such a circuit is not enough on its own.
    Just as any $Z$ error on any of the four data qubits touched propagates to the single ancilla, an $X$ error on the single ancilla can propagate to up to four data qubits.
    The fault-tolerance condition given in Definition~\ref{def:FT} demands that a single error can not propagate to more than one error on the data. 
    This is violated by the propagation of an $X$ error that occurs just before the second CNOT connecting to the data, illustrated in Figure~\ref{fig:bad_prop}.

    The insight of Reichardt and Chao was that the only error which propagates non-fault-tolerantly is the marked one, and that it is easy to check if such an error has occurred with a circuit like the one in Figure~\ref{fig:syn_flagged}.
    Since \emph{the circuit is Clifford}, the propagation of the error can be classically corrected after measuring the flag qubit (and before the next non-Clifford gate).

    The fact that we are attempting to identify the location of an error should be suggestive, despite the fact that we are trying to identify the location of the error in time instead of in space. 
    In earlier work, we showed that by treating time-like regions of a qubit as classical bits, measuring a weight-$w$ stabilizer can be made fault tolerant to distance $2t+1$ using only $O(t^2 \log w)$ extra ancilla qubits.
    In the rest of this paper, we generalize this. 
    Instead of dividing the qubit into time-like regions which we consider as classical bits, we consider space-time stabilizers of a circuit. 
    In our previous work, the space-time stabilizers we were effectively measuring were very simple, and could be measured by pairs of CNOTs. 
    In this work, we present an only slightly more complicated generating set of space-time stabilizers. 
    We show that the measurement of these stabilizers is sufficient to ensure the fault-tolerance of an arbitrary Clifford circuit and, building upon our previous work~\cite{anker2025compressed}, we show that this measurement schedule can be compressed to minimize the extra resources required. 
    \section{Constructing Flagged Clifford Circuits}\label{sec:flagged_cliffords}
    We build our flag circuit using two types of gadgets -- one that catches $X$ errors and one that catches $Z$ errors, similar to how we can focus on a single type of error in Section~\ref{sec:general_flags}. 
    For the sake of presentation, we first handle $X$ errors in detail, before extending to $Z$ errors.
    We begin by considering a Clifford circuit as a network of CNOTs, postponing the single-qubit Cliffords that allow us to capture general Clifford circuits until Section~\ref{sec:1qclif}. 
    For convenience, we denote measurement outcomes by $0$ and $1$, rather than $+1$ and $-1$.
    \subsection{Main Gadget}\label{sec:main_gadget}
    For each CNOT we introduce a single gadget composed of two flag (ancilla) qubits and five CNOTs. The form of the gadget is given in Figure~\ref{fig:main_gadget}.
    \begin{figure}[h!]
        \includegraphics[width=\columnwidth]{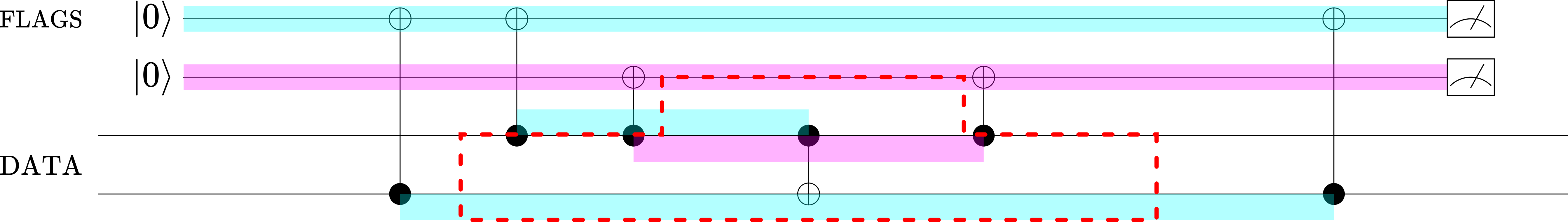}
        \caption{\label{fig:main_gadget}The main gadget we consider. The flags and the regions in which they detect ($X$) errors are color-coded. In the dashed red line we have outlined the ``bulk'' region with which we will concern ourselves.}
    \end{figure}
    The relevant property of this gadget is given by the following lemma.
    \begin{lemma}\label{lem:main_gadget_lemma}
        Any error consisting solely of $X$ errors in the red-outlined bulk region either
        \begin{enumerate}
            \item Causes at least one of the flags to be measured as $1$ or,
            \item Propagates to the identity on the data.
        \end{enumerate}
    \end{lemma}
    \begin{proof}
        We divide the bulk region into three subregions, $A$, $B$, and $C$, where $A$ is the region covered by blue and magenta flags, $B$ is the region covered only by magenta, and $C$ is the region covered only by blue. It is enough to consider each error of the form $A^{a}B^{b}C^{c}$ for $a, b, c \in {0, 1}$ where by $A^0$ we mean an even number of $X$ errors occur in region $A$, by $A^1$ an odd number, and similarly for the other terms. This is because the parity of errors in these regions define the logically distinct regions for an error to occur -- any error and any flag pattern can be explained by a product of errors in these regions. The flag pattern is then given by \[ He :=
            \begin{pmatrix}
                1 & 0 & 1\\
                1 & 1 & 0
            \end{pmatrix}\begin{pmatrix}
            a\\b\\c
            \end{pmatrix}
        \] where $H$ and $e$ are both over $\mathbb Z_2$.
        We can then just consider the cases where all flags are measured as $0$, i.e. the null-space of $H$. This is clearly spanned by $\{\begin{pmatrix}0\\0\\0\end{pmatrix}, \begin{pmatrix}
            1 \\ 1 \\ 1
        \end{pmatrix}\}$. 
        In both cases the error defined by $e$ propagates to the identity $II$ on the data, completing the proof.
    \end{proof}
    This gives us a na\"ive way to capture bulk errors given a larger circuit.
    \begin{cor}\label{cor:bulk_coverage}
        Given a network of $n$ CNOTs, adding one such gadget to each CNOT implies that any error consisting solely of $X$ errors in the red-outlined bulk regions either
        \begin{enumerate}
            \item Causes at least one of the flags to be measured as $1$ or,
            \item Propagates to the identity on the data.
        \end{enumerate}
        The total cost is $2n$ flag qubits and $5n$ CNOTs.
    \end{cor}

    For now, we make no claims regarding data errors outside the bulk region, measurement errors on the flag qubits, or errors on the flags that may propagate to the data.
    The first two issues will be addressed in Section~\ref{sec:connecting}, while the third will be addressed in Section~\ref{sec:meta}. 
    \subsection{Full Gadget}
    We now handle $Z$ errors. 
    The method is straightforward: we just apply the same gadget to the same area with the CNOTs' (including the data CNOT) controls and targets reversed and with the ancillae prepared in the $\ket{+}$ state instead of the $\ket{0}$ state.
    The new gadget for $Z$ errors and the resulting total gadget are shown in Figure~\ref{fig:z_gadget}.
    The result of adding to our gadget this way is that we can extend~\ref{cor:bulk_coverage} to include $Z$ errors as well, which directly implies we can also handle $Y$ errors since our gadgets for $X$ and $Z$ are independent. To be explicit:
    \begin{cor}\label{cor:bulk_coverage_xz}
        Given a network of $n$ CNOTs, adding one such gadget to each CNOT implies that any Pauli error with support only in the red-outlined bulk regions either
        \begin{enumerate}
            \item Causes at least one of the flags to be measured as $1$ or,
            \item Propagates to the identity on the data.
        \end{enumerate}
        The total cost is $2n$ flag qubits and $5n$ CNOTs.
    \end{cor}
    \begin{figure}
        \centering
        \subfloat[The gadget which catches $Z$ errors around a data CNOT.]{
            \includegraphics[width=\columnwidth]{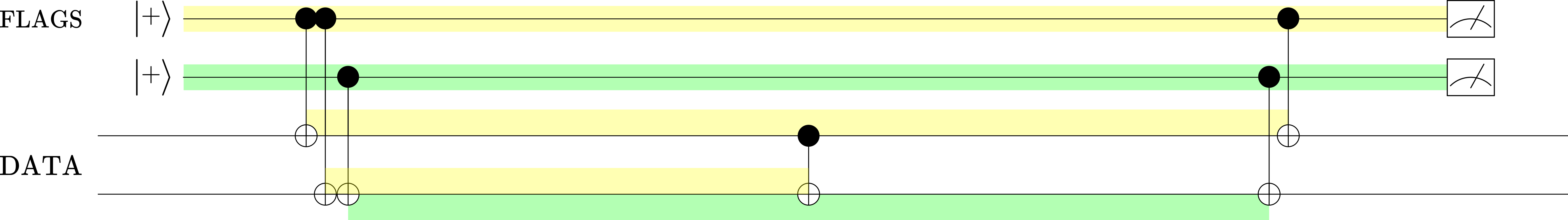}
        }\\[1em]
        \subfloat[The entire resulting gadget.]{
            \includegraphics[width=\columnwidth]{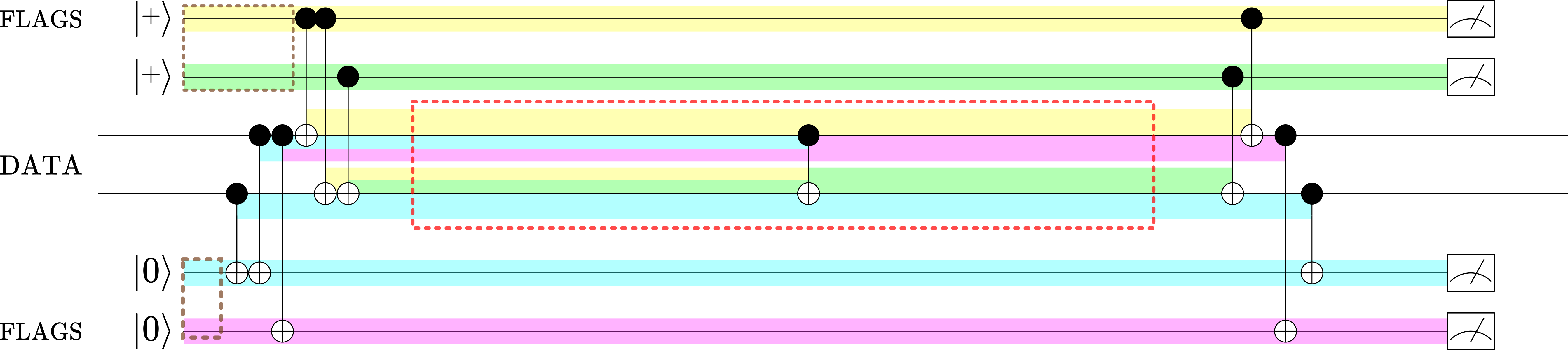}
        }
        \caption{\label{fig:z_gadget}Note that the $Z$-error gadget is the same as the $X$ gadget except that the role of control and target are exchanged for every CNOT (including the data CNOT) and ancillae are prepared in the $\ket +$ state instead of the $\ket 0$ state. For the full gadget, we have made the choice to place the gadget to catch $Z$ errors inside the gadget to place $X$ errors.}
    \end{figure}

    It is instructive to consider the fundamental structure of the circuit we obtain.
    We can label the space-time locations surrounding a CNOT as in Figure~\ref{fig:space_time_qubits}. 
    We can then consider (a generating set of) the stabilizers of this circuit:
    \begin{align}
        Z_1Z_2,Z_3Z_2Z_4,\\
        X_1X_2X_4, X_3X_4.
    \end{align}
    The space-time stabilizers of a circuit can be thought of simply as constraints on the Pauli frames across input/output time-slices of a circuit~\cite{delfosse2023spacetime}.
    The gadget we propose for detecting $X$ errors measures the two $Z$ stabilizers, and vice-versa.
    By measuring the stabilizers associated with each CNOT we \emph{enforce} the input-output conditions.
    This is an alternative way to see that any $X$ error is caught by using the gadget with flags prepared in the $\ket 0 $ state.

    As a consequence of the fact that we are measuring stabilizers and the fact that we measure one `inside' the other, errors on the flag qubits in the regions outlined in brown in Figure~\ref{fig:z_gadget} propagate trivially both to the data and to the flag qubits prepared in the opposite basis.
    
    \begin{figure}
    \centering
    \includegraphics[width=0.3\textwidth]{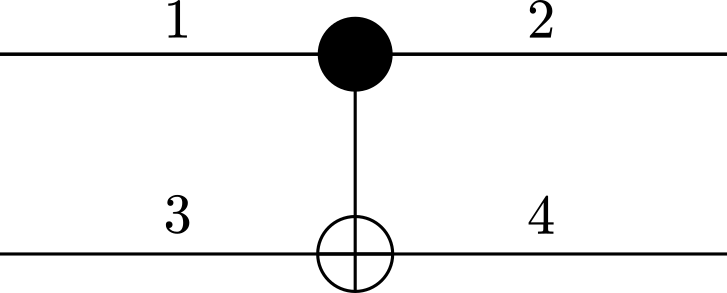}
    \caption{\label{fig:space_time_qubits}The space-time locations associated with a CNOT. We will use these locations to define space-time stabilizers.}
    \end{figure}
    \subsection{Connecting Gadgets and Handling Measurement Errors}\label{sec:connecting}
    We now extend Corollary~\ref{cor:bulk_coverage_xz} to produce a network of CNOTs which is flagged in the bulk. 
    For clarity we have limited to consider $X$ errors and gadget to detect $X$ errors, but the same statement applies for Pauli errors.
    \begin{figure*}[ht]
        \includegraphics[width=0.45\textwidth, valign=c]{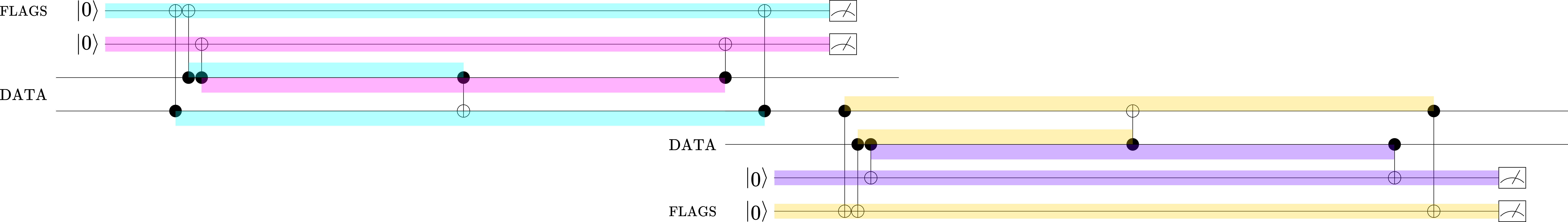}\hfill
        \includegraphics[width=0.03\textwidth, valign=c]{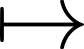}\hfill
        \includegraphics[width=0.45\textwidth, valign=c]{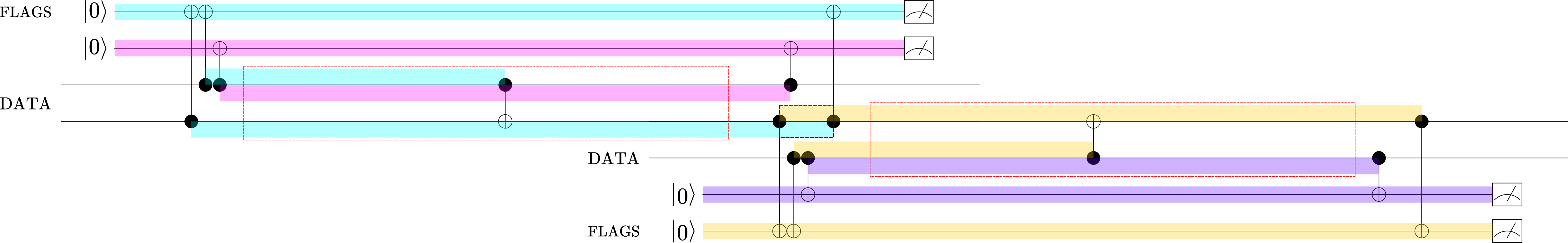}
        \caption{\label{fig:overlap_map} The procedure we use to ensure that adjacent gadgets do not leave an unprotected gap. We have only shown the $X$ gadgets for the sake of visual clarity, but the full gadgets should be overlapped in the same manner.}
    \end{figure*}
    \begin{lemma}
        Given a network of $n$ CNOTs, adding one $X$ gadget to each CNOT, and overlapping adjacent gadgets as in Figure~\ref{fig:overlap_map} implies that any error of weight-$q$ consisting solely of $X$ errors is equivalent to a weight-$q$ error consisting solely of $X$ errors each in exactly one of the red-outlined bulk regions plus at most $q$ measurement errors.
    \end{lemma}
    \begin{proof}
        We simply examine the blue outlined region. 
        A weight one error in the blue region is equivalent to a weight one error in one of the two neighboring bulk regions accompanied by a measurement error. 
        This holds for all ten possible orientations in which the two gadgets can neighbor each other.
    \end{proof}
    
    Performing the same procedure to overlap the edges of the gadget detecting both $X$ and $Z$ errors we can conclude that any Pauli error of weight-$q$ is equivalent to a Pauli error of weight-$q$ with each component in exactly one bulk region plus at most $q$ $Z$-basis measurement errors and $q$ $X$-basis measurement errors.

    Diagrammatically, we can now represent our gadgets as a collection of red-outlined bulk regions, where errors are flagged properly, and blue-outlined boundary regions, where data errors produce measurement errors as well as being flagged, as in Figure~\ref{fig:diagram}.
    \begin{figure}
        \includegraphics[width=\columnwidth]{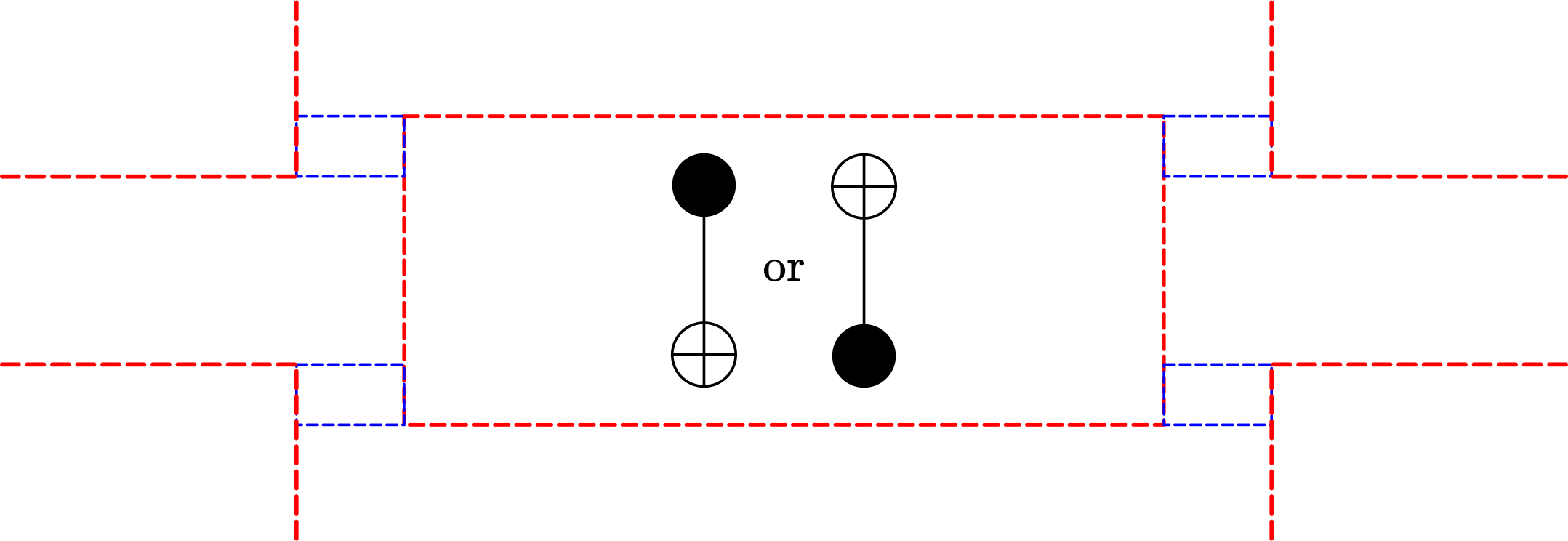}
        \caption{\label{fig:diagram}A diagrammatic representation of the bulk and boundary regions, with the bulk region protecting each data CNOT in red and the boundary in blue.}
    \end{figure}
    We now repeat the application of each gadget to each CNOT in space, simultaneously handling measurement errors and errors in the boundary regions.
    \begin{figure}
        \includegraphics[width=\columnwidth]{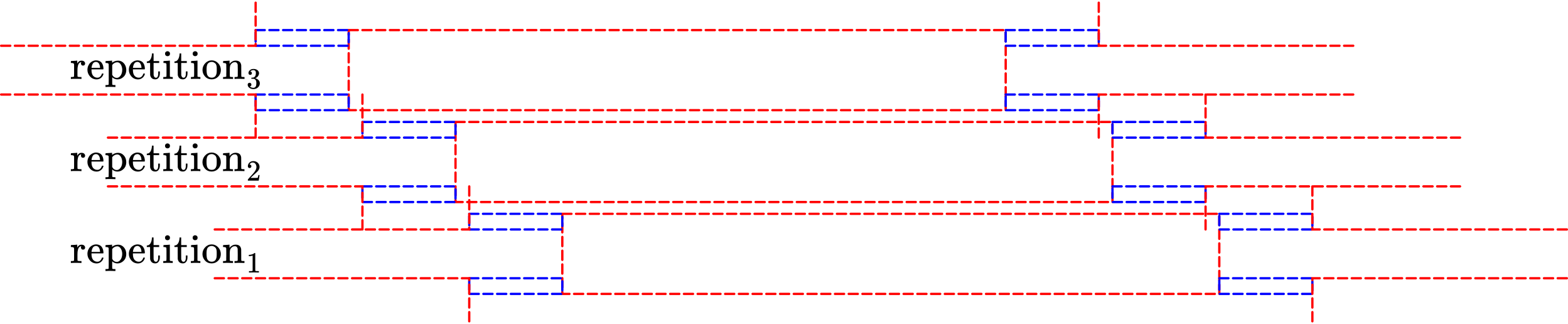}
        \caption{\label{fig:repetition} A diagrammatic representation of our repetition procedure. Each set of dotted lines represents one space-like repetition of the base flag gadget (i.e. a set of 4 flag ancillae and 10 CNOTs) attached to its neighbors. Crucially, no single-qubit error can be in more than one boundary region.}
    \end{figure}
    \begin{lemma}\label{lem:rep}
        Flagging a network of CNOTs with $2t$ repetitions of our gadget applied to each CNOT in the manner of Figure~\ref{fig:repetition} so that boundary regions of each repetition are disjoint is sufficient to ensure that any data error with support between the first and last data CNOTs either does not propagate or produces a nontrivial flag pattern even when accompanied by $m$ measurement errors, provided that $q + m < 2t$ where $q$ is the number of single-qubit errors in a boundary region.
    \end{lemma}
    \begin{proof}
        We first convert all $q$ data errors in the boundary into data errors in the bulk, producing at most an effective $q + m$ measurement errors for each basis. Now each single-qubit error is in the bulk of exactly one gadget of each type per repetition. There is at least one $X$($Z$) repetition without any (effective) measurement errors since $q + m < 2t$ by assumption. Considering this repetition, we can partition the $Z$($X$) component of the error into subsets where subset $i$ consists of all errors in the region flagged by gadget $i$. By Lemma~\ref{lem:main_gadget_lemma} we know that for each such subset either there is no propagation, or the flag pattern is non-trivial. Since flag patterns from distinct gadgets do not affect one another, and errors outside the bulk and boundary do not change the flag pattern, this suffices to prove the lemma.
    \end{proof}
    \subsection{Single-qubit Cliffords}\label{sec:1qclif}
    We now generalize from the case of a network of CNOTs to a general Clifford circuit.
    We continue to model our Clifford circuit as a network of CNOTs, only now each CNOT is followed by a single-qubit Clifford gate both on the control and the target (a product of $H, S, X, Y, Z$). Consequently, each gadget will now cover a CNOT, as well as the two (possibly trivial) Clifford gates which occur after it. We also allow for single-qubit Cliffords on the data before any CNOTs to fully capture a general Clifford circuit.
    \begin{figure}
        \includegraphics[width=\columnwidth]{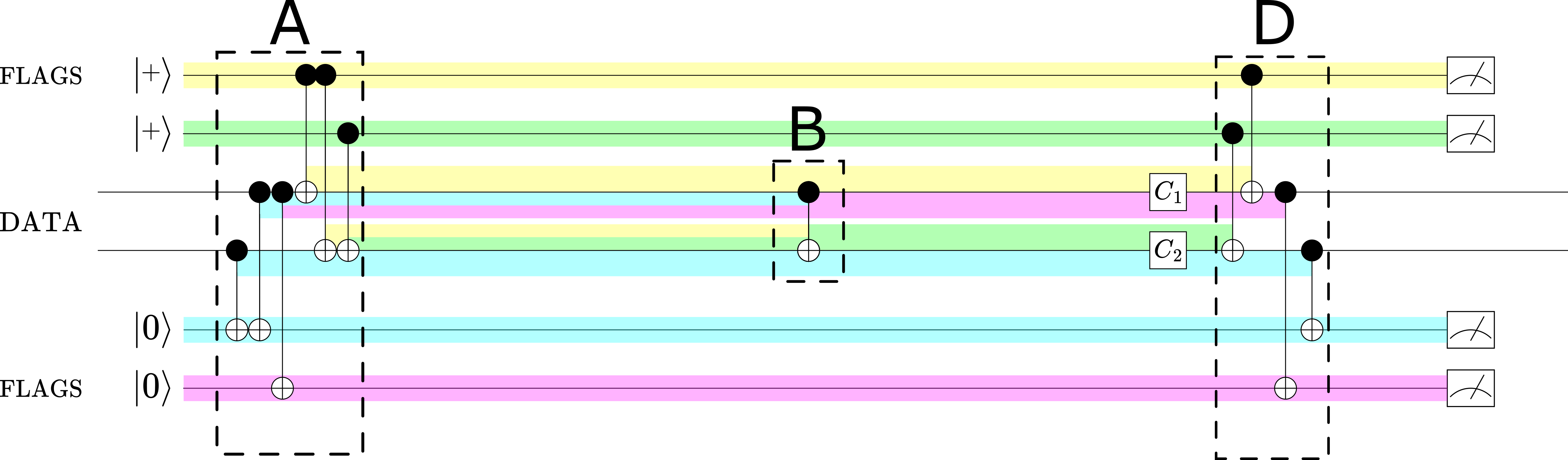}
        \caption{\label{fig:clif_boxes} The groups of gates we are interested in when we introduce single-qubit Cliffords.}
    \end{figure}
    
    Fortunately, the generalization is simple. In Figure~\ref{fig:clif_boxes} we have marked the locations of interest. When $C_1 = C_2= I$, we have implicitly been using the fact that $ABD = B$, so that we have avoided measuring any logical information, but any error in the bulk is still flagged. Equivalently we have been measuring space-time stabilizers. We will preserve this when we adjust for the fact that $C_1, C_2$ may be non-trivial, in essence by considering the stabilizers of this new subcircuit.
    \begin{lemma}
        Replacing $D$ by $D' = (C_1 \otimes C_2)^\dagger D (C_1 \otimes C_2)$ is sufficient for any error in the bulk to be flagged and for $AB(C_1 \otimes C_2)D'$ to equal $B(C_1 \otimes C_2)$.
    \end{lemma}
    \begin{proof}
        Any error before $C_1 \otimes C_2$, upon being commuted past $C_1 \otimes C_2$ and $(C_1 \otimes C_2)^\dagger$ act the same upon the flags as it would were $D'$ to equal $D$ and $C_1, C_2$ to be identity. Any Pauli error between $C_1\otimes C_2$ and $(C_1 \otimes C_2)^\dagger$, once being commuted past $(C_1 \otimes C_2)^\dagger$ will become some other Pauli error, since $(C_1 \otimes C_2)^\dagger$ is a Clifford. By the construction of the gadget, any product of $X, Z$ (i.e. any Pauli) will be caught. Any error occurring after $C_1 \otimes C_2$ acts on $D'$ as if it were $D$.

        Finally, $AB(C_1 \otimes C_2)D' = AB(C_1 \otimes C_2)(C_1 \otimes C_2)^\dagger D (C_1 \otimes C_2) = ABD (C_1 \otimes C_2) = B (C_1 \otimes C_2)$ since $ABD = B$.
    \end{proof}
    
    It is straightforward to observe that after modifying the flag gadget in this way and overlapping it with its neighbors, it is still the case that any error can be made to have support solely in the bulk regions at the price of at most one effective measurement error in each basis.
    
    The consequence that the boundaries of overlapped flag gadgets do not commute is only that some flags detect errors on other flags.
    In some sense this is by design -- if an error ends up on the data, no matter where it came from we need to detect it.
    The worry is only that one error propagates to the data multiple times.
    If, however, we only use a single gadget per flag qubit, errors on the flag qubits only propagate to the data once.
    This leads us to a Corollary providing one way to implement a Clifford circuit fault tolerantly.
    \begin{cor}[Non-compressed Fault Tolerance]
        Flagging a circuit consisting of $n$ CNOTs and up to $2(n + 1)$ single qubit Clifford gates with one gadget per CNOT allows for any Pauli error in the union of the bulk regions of each gadget to be fault tolerantly corrected. The cost of such a procedur is $4n$ ancillae and $10$ additional CNOTs.
    \end{cor}
    
    \subsection{Compression by a Classical Code}\label{sec:compression}
    Up to this point, we have focused solely on ensuring that every propagating error is flagged without considering resource efficiency. A na\"ive application of one gadget per two-qubit gate requires $2n$ ancilla qubits and $5n$ two-qubit gates to protect against the propagation of $X$ errors on the data.
    However, we can reduce this overhead by applying a classical code.
    The key insight is that measuring each flag qubit is equivalent to measuring a space-time stabilizer, as described by Gottesman~\cite{gottesman2022perspective}.
    Recognizing this allows us to compress the number of measurements, and thus the number of ancilla qubits, by leveraging our previous work on minimizing the measurements needed for fault-tolerant error correction~\cite{anker2025compressed}

    Before describing the compression procedure, we first outline some of the traps we will have to avoid.
    We will end up measuring the product of many gadgets, the precise meaning of which we will describe shortly but which involves connecting one flag qubit to many data qubits.
    A natural fear is that a low-weight error on the flag qubits can propagate to the data many times. 
    Although it is possible to argue that the final error upon propagating through the circuit is low-weight, the more pressing issue is that the flag pattern produced may be identical to the flag pattern produced for different data error.
    We observe that if each error on a flag qubit to propagate only to a single data location, this issue is obviated, since we already handle low-weight errors on the data by design.
    We will show in Section~\ref{sec:meta} that we can gaurantee this up to a known Pauli correction.

    \begin{figure*}[ht]
        \includegraphics[width=\textwidth]{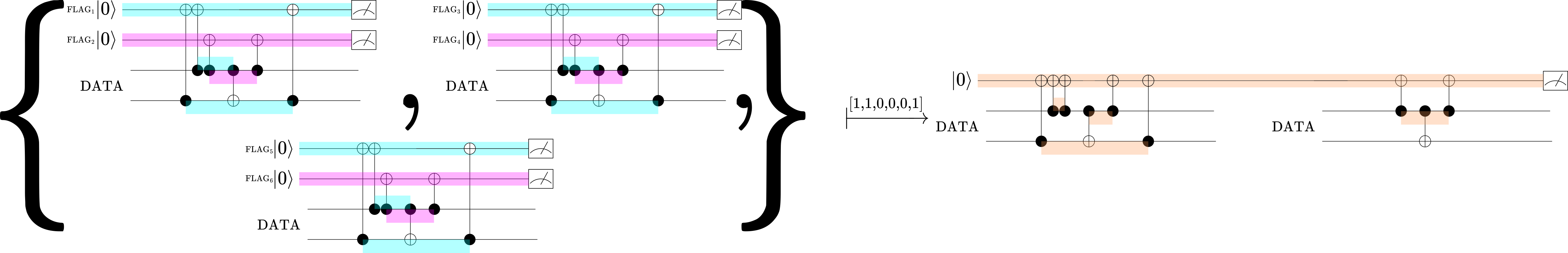}
        \caption{\label{fig:gadget_prod} Given the six measurement results that constitute the gadgets protecting two data CNOTs and the parity check $[1, 1, 0, 0, 0, 1]$ we combine the gadgets on the left to form the gadget on the right. One can verify that taking the three gadgets separately and adding measurements results $1, 2$ and $6$ produces the same result as the circuit on the right for any data error outside of the boundary regions. Note that the three sets of data qubits labeled on the left may refer to physically distinct data qubits, or may not.}
    \end{figure*}
    Formally, consider the $X$-type flag pattern defined by each measurement result in repetition $i$, which we call $f_i$. 
    \begin{lemma}
        If $s$ data errors are suffered, then $|f_i| \leq 2s$, where $|\cdot|$ refers to the Hamming weight.
    \end{lemma}
    \begin{proof}
        This follows from the triangle inequality, since in the bulk each error is in the support of two flagged regions, and where two gadgets overlap there are at most two flagged regions.
    \end{proof}
    Since $|f_i|$ is close to the all-zeros string, which is a codeword in every linear block code, we can identify the positions of the ones, and thus the flag pattern, by performing parity checks on it. 
    Specifically, let $P$ be the parity-check matrix of a classical code with distance $2r + 1$.
    Then $f_i$ can be viewed as the zero codeword corrupted by an error with bitstring $f_i$.
    By computing the syndrome $P(\vec{0} + f_i) = Pf_i$, we can recover the error $f_i$ as long as its weight is at most $r$ (i.e., up to half the minimum distance).

    Of course, $f_i$ is simply the result of applying some set of parity checks (up to single-qubit Cliffords) to the space-time volume of our circuit -- therefore we can write $f_i$ as $He$ for some $H$. Then $Pf_i = PHe$. So instead of implementing $H$ and then doing parity checks on the resulting classical flag pattern, we implement $PH$, which combines the parity checks on the errors and on the flag pattern. Implementing $PH$ just corresponds to taking products of measurements of single flag gadgets, which just corresponds to attaching two or more gadgets to the same ancilla before measuring the ancilla as in Figure~\ref{fig:gadget_prod}. 
    
    After this multiplication, it may be the case that two CNOTs become redundant, in that multiplying them together yields identity, as is the case in Figure~\ref{fig:gadget_prod}. 
    Since this happens only in the boundaries, and we handle errors in the boundaries by considering them as errors in the bulk accompanied by measurement errors, removing these redundant gates still allows for fault tolerance.

    It is notable that now errors in the boundary regions are equivalent to errors in the bulk only up to a number of measurement errors which may be more than one. However, the proof of Lemma~\ref{lem:rep} only assumes that there is at least one repetition lacking both measurement errors and errors in the boundaries, which still holds.

    Implementing this strategy for every repetition then allows us to deduce the flag pattern, and hence the data errors, while remaining resistant to measurement errors and effective measurement errors by Lemma~\ref{lem:rep}. If $P$ has $m$ rows then the number of ancilla qubits is just $2t m$, assuming we repeat $2t$ times. For the BCH code \cite{hocquenghem1959bch, bose1960bch}, this is $O(2t^2 \log(2n))$, where $n$ is the number of two-qubit gates in the Clifford circuit we wish to flag.

    The new problem introduced is that errors on the flag qubits can propagate to the data many times.
    In fact, this is only a problem for the flag gadgets of the opposite basis.
    We formalize this with the following observation.
    \begin{lemma}
        The propagation of a weight-$1$ error on a flag qubit when restricted to the data is equivalent to an error somewhere on the data of weight at most $1$.
    \end{lemma}
    \begin{proof}
        Any given flag qubit measures the product of some set of space-time stabilizers $g_1, \ldots, g_n$. 
        An error $e$ on this flag qubit then propagates to $g_i' g_{i + 1}\ldots g_n$ on the data, where $g_i'$ is obtained from $g_i$ by replacing one of its terms with identity.
        Then up to space-time stabilizers, the propagation of $e$ is equal just to $g_i' g_i$, which is at most weight $1$, up to the stabilizer $g_i$, since $g_i$ is weight at most $3$ and $g_i'$ is a subset of $g_i$.
    \end{proof}
    Since data errors are already flagged for, it appears as if the construction given up to here is fault tolerant.
    Unfortunately, because of the fact that we overlap gadgets, an error on a flag qubit prepared in the $\ket 0$($\ket +$) basis can propagate to an unbounded number of flag qubits prepared in the $\ket +$($\ket 0$) basis.
    Effectively, error on flag qubits in one basis can appear as high-weight errors on the data to flag qubits in the other basis, surpassing the distance of the code that we use to compress the measurements.
    In the next section we solve this problem by ensuring that errors from the flag qubits can only propagate to a weight-$1$ error on the data before being corrected.

    \subsection{Meta-Flags and $Z$ errors}\label{sec:meta}
    So far, we have only proven that the suggested flag gadget construction is sufficient to handle data errors and measurement errors. However, $Z$ errors on the gadgets which catch $X$ errors can propagate to multiple $Z$ errors on the data after we compress the set of gadgets by a classical code and vice versa (if they are not compressed, an error on a flag qubit propagates at most once to the data, and is hence caught by other flag gadgets). For this reason the construction up to this point is not fault tolerant. We can, however, add flag gadgets to our flag qubits, so that any error on the flag qubits is approximately localized. To do this, we follow our previous prescription for flagged syndrome extraction~\cite{anker2024flag} nearly exactly.
    
    In our previous work, we have shown how to build flag gadgets protecting a single ancilla qubit connected to many data qubits so that any error on the ancilla qubit propagates to at most one error on the data qubit up to corrections based upon the flag pattern. We also have shown how to flag multiple ancilla qubits as if they were one larger ancilla qubit. This is precisely what we need to do in this case, in that the connected space-time stabilizers (what we have been calling flags up to this point) each take the role of one of the ancilla qubits. However, we also improve upon our previous work by making an observation about the form of the meta-flags.

    To meta-flag the flags, we use the fact that each flag qubit prepared in the $Z$ basis has only controls for each CNOT on it, meaning only $X$ errors will propagate from the flag qubit (the prescription for flagging flag qubits prepared in the $X$ basis follows readily from the prescription for the $Z$ basis case, and we omit it). 
    Therefore, we can take the region between each pair of CNOTs as a location for errors again. When constructing our flag gadget, each $1$ in a parity check corresponds to multiplying the primitive gadget composed of two CNOTs surrounding one CNOT connected to the data. This is illustrated in Figure~\ref{fig:fix}. It is notable that since the primitive flag gadgets we connect to the data measure stabilizers, any error from the flag qubit before any CNOTs are connected, propagates to a stabilizer.
    \begin{figure*}
        \includegraphics[width=\textwidth]{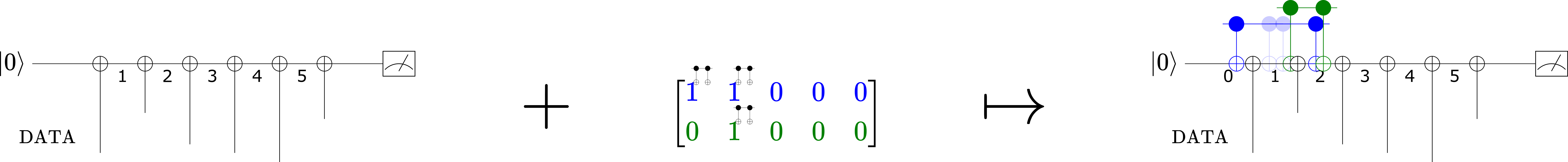}
        \caption{\label{fig:fix} On the left we have a flag qubit with the locations we consider numbered (data qubits it connects to not shown). In the center an example parity-check matrix is shown, with the gadgets (top right) associated to the protection of each location marked with a $1$. The right panel shows the product of these gadgets, color-coded by parity check. Redundant CNOTs which can be canceled are displayed translucent.}
    \end{figure*}

    Analogous to how we consider each location of the circuit as a location to flag above, we simply consider each location on any flag qubit as a location to meta-flag. We now can consider errors from the meta-flags that propagate to the flag qubits. 
    \begin{lemma}\label{lem:meta-prop}
        Any error from a meta-flag propagates to at most one flag qubit, and not at all to the data.
    \end{lemma}
    \begin{proof}
        Every meta-flag qubit touches any flag qubit an even number of times, by construction (each primitive gadget is composed of $2$ CNOTs and multiplying any number of such gadgets will produce a gadget of even weight), meaning that CNOTs from flag qubit $i$ occur before CNOTs from flag qubit $i + k$, $k> 0$, no error can propagate an odd number of times to both $i$ and $i + k$. Only $Z$ errors propagate from meta-flags to flags, and $Z$ errors do not propagate to the data from flag qubits.
    \end{proof}
    
    This can easily be seen to ensure that any two errors $e$ and $e'$ differ by a space-time correction of weight at most $\min(|e|, |e'|)$. Of course, this is not necessarily sufficient for fault-tolerance, since even an error of low weight as it occurs can propagate to an error of high weight. But we have already designed gadgets that ensure any low-weight error on the data is flagged sufficiently for fault-tolerance. Therefore, by adding meta-flags, we convert low-weight errors on the flag gadgets into low-weight errors on the data, which are caught by some set of flag gadgets by design.

    We would like to emphasize the distinction to our previous work. 
    In Table~\ref{table:comparison} we outline the correspondence between the gadgets we use in both works. In particular, we note that we end up with one less level of protection than in our previous work, but still claim the same level of fault-tolerance (this work capturing a general class of circuits that can be specialized to syndrome extraction to produce a better result). The difference is because in this work we have considered the construction of flag gadgets by multiplying primitive parity checks and compressing a syndrome. In our previous work, the difference in construction did not make it provable that errors on flags would not propagate to more than one of the syndrome extraction qubits, meaning that we ran the risk of low-weight flag errors producing effective high-weight syndrome errors, leading to a non-fault-tolerant error correction. In this work, due to Lemma~\ref{lem:meta-prop} we can see that this is not the case.
    \begin{table}[h!]
        \centering
        \begin{tabular}{l|r}
            \hline
            Previous work \cite{anker2024flag} & This work \\
            \hline
            Stabilizer measurement & Product of the gadgets\\
                                     & defined in Figure~\ref{fig:main_gadget}\\\\
            Flags protecting multiple&Meta-flags\\
            stabilizer measurements  &(this section) \\\\
            Meta-flags & Unnecessary
        \end{tabular}
        \caption{The correspondence between gadgets used in this work and our previous work \cite{anker2024flag}.}
        \label{table:comparison}
    \end{table}

    \subsection{Boundary Regions}\label{sec:boundary_gadget}
    Finally, we need to handle errors that occur on the data directly before the application of any flags. These errors are undetectable to our current set of flags by definition, but can propagate to high weight errors on the data through the Clifford circuit. In this section we define a procedure to ensure that errors before the Clifford circuit are either corrected or flagged. In brief, we insert error correction inside a set of flags on the boundary as illustrated in Figure~\ref{fig:ec_gadget}.
    \begin{figure}[h]
        \includegraphics[width=\columnwidth]{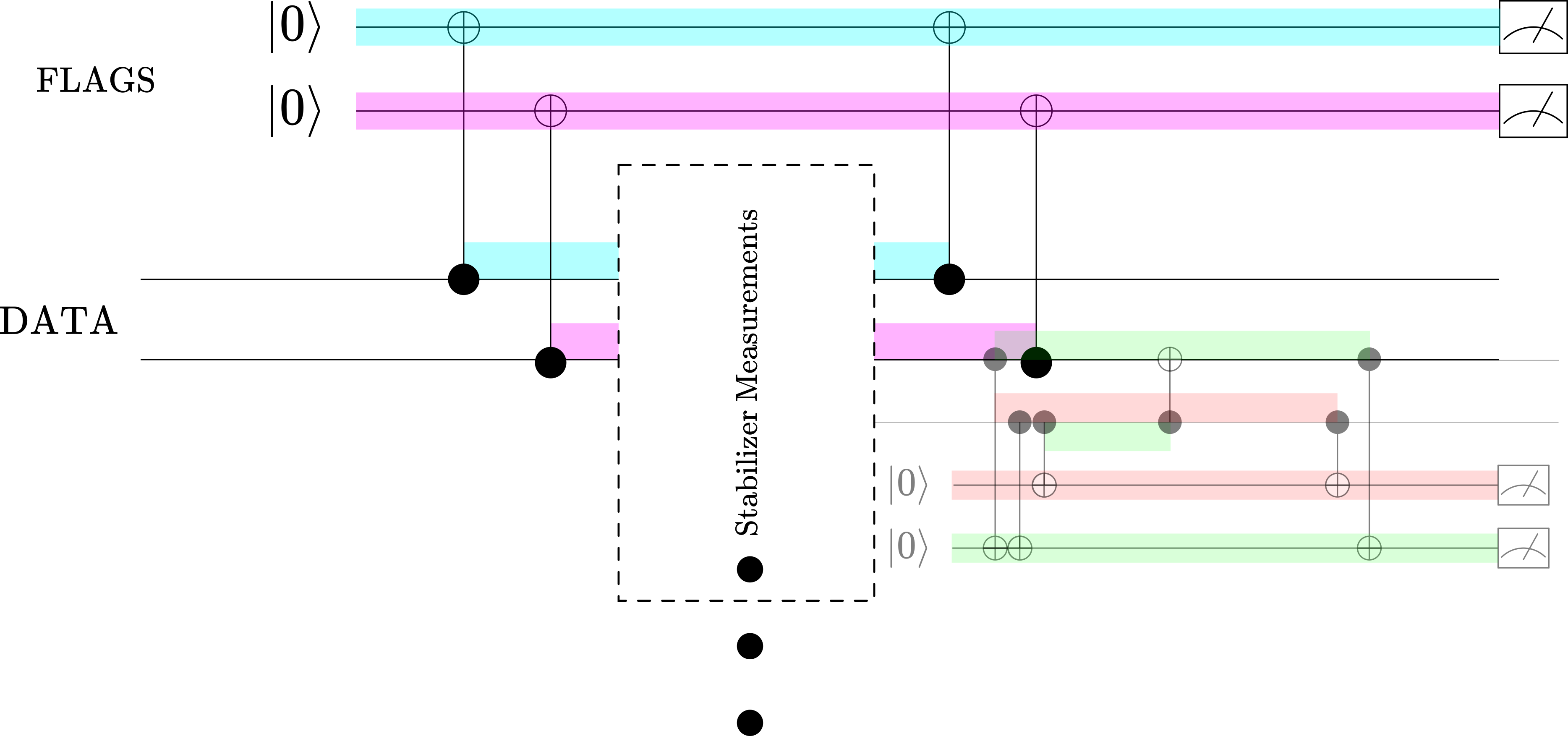}
        \caption{\label{fig:ec_gadget} A diagrammatic representation of the way in which we ensure errors before flags do not propagate to high-weight errors. Flags for $Z$-type errors are omitted but are applied in a corresponding manner -- CNOTs from flags prepared in the $\ket +$ state to the data. Partially transparent we have an example of a gadget protecting one of the data CNOTs and how it interacts with the gadget covering EC.}
    \end{figure}
    This allows us to prove similar results about detectability as previously, and then to apply the same machinery as in Section~\ref{sec:compression}, only with slightly more locations.
    \begin{lemma}
        Building a gadget as in Figure~\ref{fig:ec_gadget} with a set of stabilizer measurements defining a measurement scheme which is fault tolerant to distance $t$ means that any error of weight at most $t$ before stabilizer measurements is detected, as well as any error in the colored regions.
    \end{lemma}
    \begin{proof}
        This follows directly, in that a fault-tolerant error correction scheme necessarily detects data errors of up to weight-$t$, and the non-trivial errors in the colored regions are caught by the flags by design. Logical operators and stabilizers are unaffected by the error correction procedure, and hence the flag gadgets still only detect errors, not the logical state of the code.
    \end{proof}
    Note that both the flags and the error correction may detect some errors, and that the flags may detect errors on ancilla qubits used for the stabilizer measurements. As long as the stabilizer measurements are done fault-tolerantly, this does not impact our ability to follow the compression procedure previously described. It is also worth clarifying that the flag qubit surrounding the error correction protocol are not there to ensure the fault-tolerance of the error-correction procedure - rather they are there to give the flag gadgets on the circuit something to `hook into', or to ensure that the flags have a baseline to flag against. That is to say, after adding these flag gadgets to our construction, we follow the same prescription (on their right sides) for overlapping adjacent flag gadgets as we did for all other flag gadgets. This is illustrated partially transparent in Figure~\ref{fig:ec_gadget}.

    \begin{thm}
        Given a Clifford circuit composed of $n = O(m^2)$ CNOTs and up to $2(n + 1)$ single-qubit Cliffords which acts upon the physical qubits of an $[[m, \cdot, 2t' + 1]]$ code, it can be made fault tolerant to distance $2t + 1 $ where $t \leq t'$ using $O(t^2\log n + t^2\log(n t^2 \log n))$ ancilla qubits, $O(nt^2\log n + nt^2\log(n t^2 \log n))$ additional CNOTs and one application of fault-tolerant error correction.
    \end{thm}
    \begin{proof}
        We simply apply our compression protocol from Section~\ref{sec:compression} using the classical BCH code to the $4$ flag qubit measurements for each of the $n$ data CNOTs plus the $2m$ measurements used to flag the error correction.
    \end{proof}
    The first term in the sum counts the number of flags qubits, while the second is the number of meta-flag qubits. The term inside the log is just the number of flag qubits multiplied by an upper bound on the number of locations a single flag qubit can add.

    In fact, the stabilizer-measurement gadgets can be understood as another type of primitive gadget to multiply together to form the parity check flag gadgets. 
    This unifies the space-like and the space-time stabilizer measurements. 
    The locations on the ancilla qubits used for (code, not space-time) stabilizer measurements are then just locations that we flag with the meta-flags from Section~\ref{sec:meta}. 
    The only subtlety is that repeated rounds of syndrome extraction are (usually) necessary for fault tolerance, so if our parity-check matrix tells us to take the product of two measurements in different rounds this must be done with two measurements multiplied classically, instead of one measurement of the product.

    The BCH code suggested in Section~\ref{sec:compression} uses the fewest number of additional ancillae for this construction, but is decidedly sub-optimal in terms of the number of additional CNOTs used. 
    Each row of the BCH parity-check matrix has Hamming-weight $O(n)$ where $n$ is the number of two-qubit gates in the circuit we wish to flag. 
    Therefore, each flag-qubit uses a linear number of CNOTs. 
    This is compounded by the fact that the number of CNOTs used by the first level flags is the number of locations for the second level flags. 
    While the number of CNOTs remains linear in the original number, the constant factors can be quite high.

    This problem is alleviated, at the cost of additional ancillae for some regimes, by using a good (classical) LDPC code for both layers of flags. This leads us to our next theorem.
    \begin{thm}
        Given a Clifford circuit composed of $n$ CNOTs and up to $2(n + 1)$ single-qubit Cliffords which acts upon the $m$ physical qubits of a code of distance $2t' + 1$, provided $n = \Theta(t)$, it can be made fault tolerant to distance $2t + 1 $ where $t \leq t'$ using $O(n)$ ancilla qubits, $O(n)$ additional CNOTs and one front end application of fault-tolerant error correction.
    \end{thm}
    \begin{proof}
        We simply apply our construction using a good (classical) LDPC code with parameters $[[n, O(n), O(n)]]$. The weight of each row of the parity-check matrix is constant, meaning the number of CNOTs connecting to each flag qubit is constant. This in turn means that the number of locations for the meta-flags is proportional to the number of flag qubits, which is clearly $O(n)$. Applying a good LDPC code to the flags to meta-flag them follows the same resource analysis. Therefore each flag or meta-flag qubit has only $O(1)$ CNOTs connected to it.
    \end{proof}
    This approach offers a middle ground between the $4n$ ancilla qubits and $10n$ CNOT gates required without compression, and the $O(t^2 \log(n \log n))$ ancilla qubits with $O(nt^2\log n + nt^2\log(n t^2 \log n))$ CNOT gates achieved through BCH code compression.
    \section{Application to the $[[15, 1, 3]]$ Quantum Reed-Muller Code}\label{sec:logical_H}
    We now demonstrate a justification of our title, universal fault tolerance with non-transversal Clifford gates. First, we review the $[[15, 1, 3]]$ Quantum Reed-Muller code \cite{knill1996reedmuller}, also referred to as the 15-qubit tetrahedral or 3D color code. This is a CSS code obtained using the classical punctured Reed-Muller code. For our purposes, the most important aspect of this code is that it is the smallest error-correcting code with transversal $T$ \cite{koutsioumpas2022smallestT}. Since it is a CSS code, CNOT is also transversal -- therefore, to obtain the universal gate set of $\{H, T, \text{CNOT}\}$ fault tolerantly, it is enough to implement $H$ fault tolerantly. Although, by the Eastin-Knill theorem~\cite{eastin2009transversal}, we know that such a gate cannot be implemented transversally, by our construction this is unnecessary.

    Before we demonstrate our construction applied to this code, we first explicitly define the stabilizers and one implementation of logical $H$. The code is defined on a tetrahedron, where each cell (of any dimension) corresponds to a physical qubit -- that is, each vertex, edge, and face hosts a physical qubit, as does as the entire volume. Drawing new edges between any qubit on a $j$-cell and any qubit on a neighboring $(j+1)$-cell, the volume of the tetrahedron is then divided into four identical polytopes, each with 8 vertices and 6 faces. Each polytope defines a weight-$8$ stabilizer of $X$ type, and each face of any polytope defines a weight-$4$ stabilizer of $Z$ type. Logical $Z$ is $ZZZ$ on any of the edges of the tetrahedron, while logical $X$ is $X^{\otimes 7}$ on any face of the tetrahedron. This is summarized in Figure~\ref{fig:tetrahedral_code}.
    \begin{figure}
        \includegraphics[width=\columnwidth]{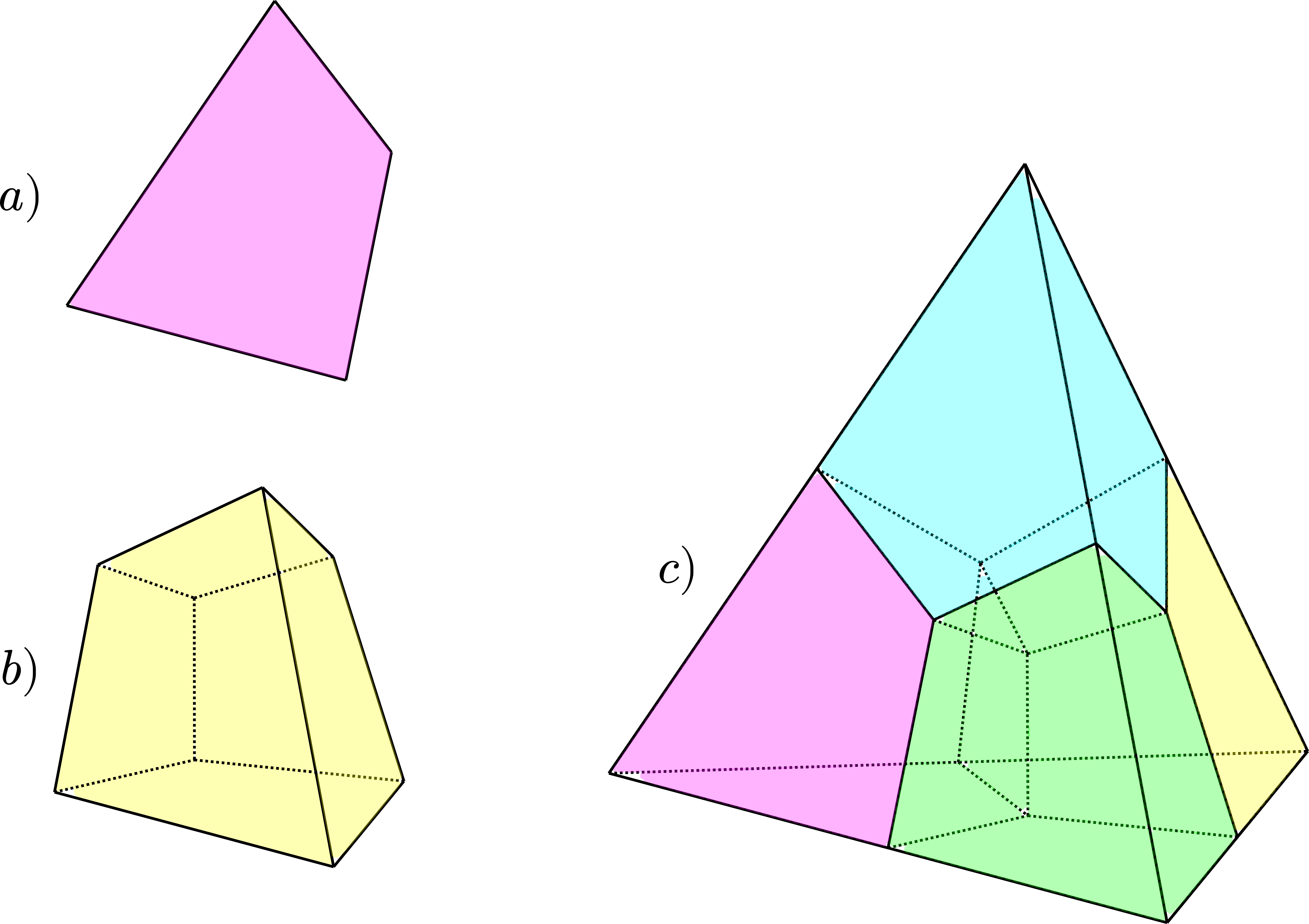}
        \caption{\label{fig:tetrahedral_code} a) a $Z$ stabilizer, b) an $X$ stabilizer, c) how they fit together to form the whole code. Qubits are on each of the 15 intersections. We have taken the figure unmodified from previous work~\cite{anker2025compressed}.}
    \end{figure}
    We number the physical qubits from 1 to 15 top-to-bottom (physically on the page), left-to-right.

    We can now define logical $H$. Calling the stabilizer group by $G$, and logical $X, Z$ respectively by $\bar X, \bar Z$, we wish to find an operation which satisfies $\langle G, \bar Z\rangle \mapsto \langle G, \bar X\rangle$ and $\langle G, \bar X\rangle \mapsto \langle G, \bar Z\rangle$, where $\langle L \rangle$ is the group generated by $L$. One implementation of the logical operation defined by this relation is given in Figure~\ref{fig:logical_H}.
    \begin{figure*}
        \includegraphics[width=\textwidth]{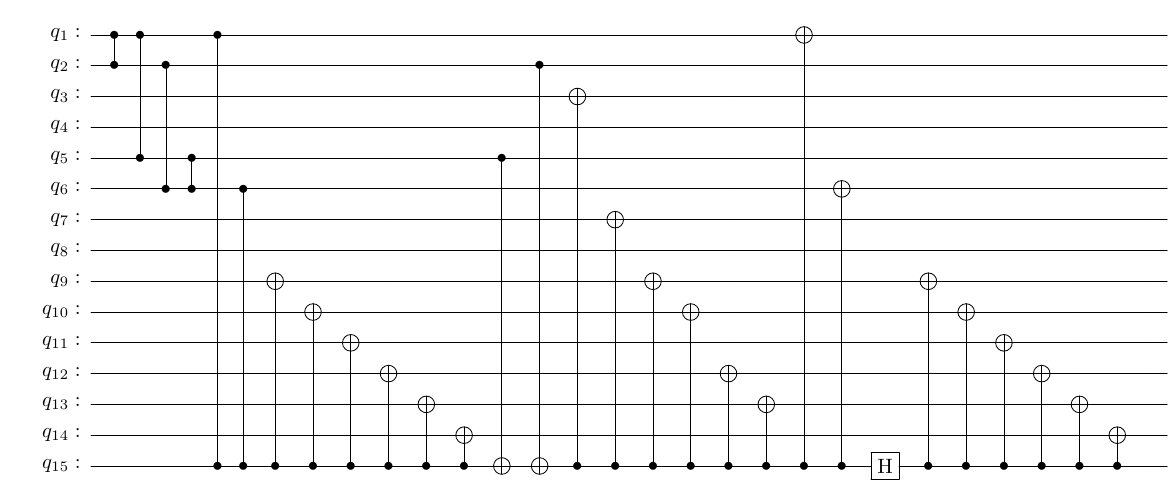}
        \caption{\label{fig:logical_H} One implementation of $\bar H$ on the $[[15, 1, 3]]$ tetrahedral code using only one- and two-qubit physical Clifford gates.}
    \end{figure*}
    
    Note that we do not claim, nor particularly believe, that this is the most efficient implementation of logical $H$. Notably, every stabilizer commutes with the circuit; i.e., this implementation of $\bar H$ centralizes $G$, rather than merely normalizing it as would be sufficient. This circuit was constructed by generating two subcircuits, $\text{PREP}_1$ and $\text{PREP}_2$, which respectively prepare the states stabilized by each generator along with $X$ or $Z$ (i.e., the logical states $\ket{\bar 0}$ and $\ket{\bar +}$) from the initial state $\ket{0}^{15}$ while simultaneously preserving the correct action of the conjugate pairs -- destabilizers associated with each $g \in G$ as well as $\bar Z$ and $\bar X$, respectively. These subcircuits were generated using Stim~\cite{gidney2021stim}. By applying $\text{PREP}_1^{-1} \text{PREP}_2$, we obtain one implementation of $\bar H$ which is essentially decoding followed by re-encoding. A naive composition of the circuits produced a 93-gate implementation, which we heuristically optimized with PyZX~\cite{kissinger2020pyzx} to reduce to 28 gates. Since our construction uses only CNOT and single-qubit gates, this becomes 36 gates after converting CZs into CNOTs conjugated by $H$, but the CNOT count (the relevant number for our construction) remains at 28.

    In some sense \emph{because} of the fact that logical operators, in particular $\bar H$, are only defined upon logical degrees of freedom, the fault-tolerance or lack-thereof of this construction is a degree of freedom we can modify. Applying our construction to this circuit is enough to ensure it is fault tolerant, and hence enough for universal logical computation. 
    A diagram of the logical computation procedure we envision is outlined in Figure~\ref{fig:computation}.
    \begin{figure}[b]
        \centering
        \includegraphics[width=\linewidth]{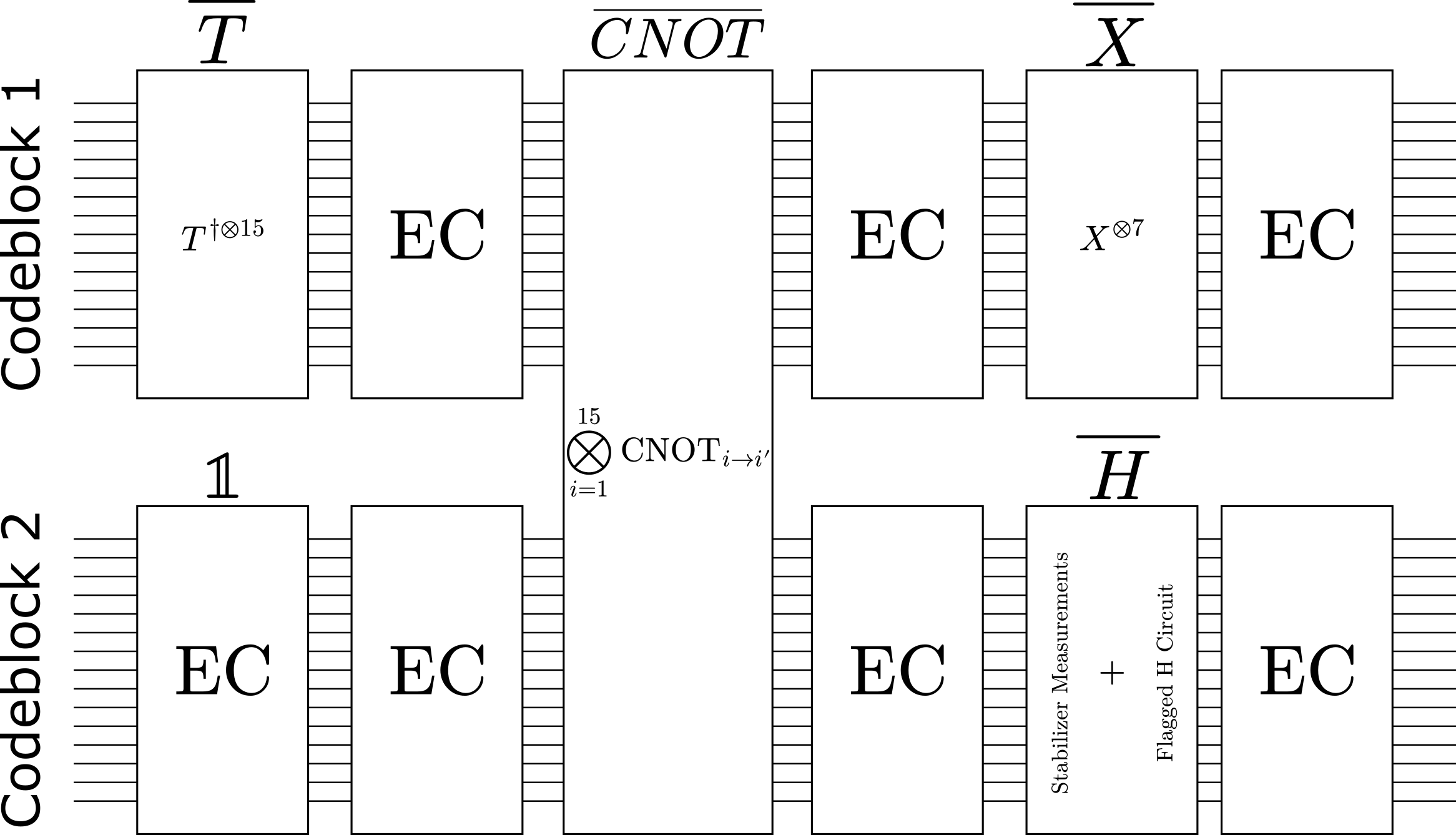}
        \caption{\label{fig:computation} A diagram of the model of computation we are proposing with the example logical circuit $T_1\text{CNOT}_{1 \rightarrow 2} X_1 H_2$. The general procedure is to alternate logical operators and fault-tolerant error correction. Note that the identity operator is usually best thought of as a round of error correction. Boxes represent logical operators, with labels for the logical operator above and the implementation in mind inset. Logical $T$ is implemented by the transversal application of $T^\dagger$, logical CNOT by the tranversal application of CNOT between pairs of qubits with the same index in each block, logical $X$ by a tensor product of $X$ operators on any of the faces of the tetrahedron, and logical $H$ by the gadget we have outlined in Section~\ref{sec:flagged_cliffords}.}
    \end{figure}
    
    We discuss the resources used by this construction in Section~\ref{sec:logical_H_resources} as well as noting some simple optimizations to reduce the resources used. We also discuss another way to use the flags gadget introduced in Section~\ref{sec:main_gadget} to reduce the logical error rate of this circuit in Appendix~\ref{sec:random_search}. This is an example of reducing the level of fault tolerance to lie somewhere between full fault tolerance and the raw circuit, while correspondingly reducing the resources needed to again lie somewhere in between. First we consider other, more standard, methods to achieve universal fault tolerance and for comparison against our method.

    \subsection{Universality based on Code Switching}
    Although the 15-qubit quantum Reed-Muller (QRM) code does not admit transversal $\bar H$, the Steane code~\cite{steane1996steanecode} does. This fact is significant because there exist relatively simple and low-overhead fault-tolerant methods~\cite{paetznick2013switch, butt2024switch, kubica2015color, bombin2015color,poulin2014color, heussen2025efficientfault} to code switch between the QRM code and the Steane code. In fact, the 15-qubit QRM and the 7-qubit Steane code are two examples of an infinite hierarchy of quantum Reed-Muller codes which at the $m$-th level admit transversal implementations of (all) gates from the $m$-th level of the Clifford hierarchy, and previous work has explored code-switching between any two levels.

    The resource estimates we desire, the number of measurements and CNOTs used, are not directly available from other work optimized for other objectives.
    Therefore in this section we follow previous work by Poulin et al.~\cite{poulin2014color} and specialize it to the case we are interested in, switching between the 7-qubit and 15-qubit QRM codes. 
    Following Poulin, we first describe how to transform the logical state $\ket{\bar \psi}_{7}$ to the logical state $\ket{\bar \psi}_{15}$, where the subscript denotes the code the state is encoded into.
    
    First, we prepare the ancilla state $\ket\Phi \propto \ket{\bar 0}_7 \ket 0 + \ket{\bar 1}_7\ket 1$, a maximally entangled state between the codestate and an additional single (physical) qubit ancilla. We then consider the stabilizers of the state $\ket{\bar \psi}_{7} \ket \Phi$. These stabilizers are generated by products of the following six forms of stabilizers:
        \begin{gather}
            g_z \otimes \bar I \otimes I\nonumber\\
            g_x \otimes \bar I \otimes I\nonumber\\
            \bar I \otimes g_z \otimes I\nonumber\\
            \bar I \otimes g_x \otimes I\nonumber\\
            \bar I \otimes\bar Z\otimes Z\nonumber\\
            \bar I \otimes\bar X\otimes X
        \end{gather}
    where $g_x$ and $g_z$ are $X$ and $Z$ stabilizers of one of the $7$-qubit codes and $\bar X, \bar Z$ are the logical operators.

    We can trivially rewrite the stabilizer generators in the following form:
    \begin{gather}
        g_z^i \otimes g_z^i \otimes I\nonumber\\
        g_x^j \otimes g_x^j \otimes I\nonumber\\
        \bar I \otimes\bar Z\otimes Z\nonumber\\
        \bar I \otimes\bar X\otimes X\nonumber\\
        g_z^k \otimes \bar I \otimes I\nonumber\\
        g_x^\ell \otimes \bar I \otimes I
\refstepcounter{equation}\tag{\theequation}\label{eqn:stabilizers}
    \end{gather}
    for $i, j, k, \ell \in \{1, 2, 3\}$. By $g_z^i \otimes g_z^i$ we mean the same stabilizer generator from the two $7$-qubit codes multiplied by one another (these will be seen to define the polytope stabilizers of the $15$ qubit code), and similarly for $g_x^j$. 
    We now can recognize the first five rows as also being stabilizers of the $15$-qubit code. The remaining weight-$4$ $x$ stabilizers are not stabilizers of the $15$ qubit code. However, they are also not required to correct errors on the $\ket{\bar \psi}_7 \ket \Phi$ state -- therefore, as shown by Poulin et al. it is enough to measure the stabilizers given, correct any errors using the syndrome defined by the measurement of the first five lines, then fix the values of the remaining stabilizers as $+1$ using their associated pure errors (or destabilizers) as a correction (gauge fixing from a 15-qubit subsystem code to the tetrahedral code we are interested in).

    In summary, to convert from $\ket{\psi}_7$ to $\ket{\psi}_{15}$ fault-tolerantly it is enough to fault-tolerantly
    \begin{enumerate}
        \item prepare the ancilla state $\ket \Phi$
        \item measure the 14 stabilizers defined in equation~\ref{eqn:stabilizers} 
        \item error correct based on the first 11 measurements, using pure errors to set the measurement results of the remaining stabilizers.
    \end{enumerate}

    The method to convert from the $15$-qubit to the $7$-qubit code is similarly simple: measure the stabilizer generators of the $15$ qubit code, correct based on the first $11$ bits, and set the last bits to $+1$ using their associated pure errors, at which point the 8 additional physical qubits, now unentangled from the $7$ qubits constituting the code, can be thrown away. 

    Therefore, one method for universality simply performs a given gate $\mathcal L \in \{\text{Clifford} + T\}$ transversally if it is transversal in the current code and switches to the complementary code, where the gate \emph{is} guaranteed to be transversal, if not.

    The cost of this can be estimated easily enough. The cost of switching from $7$ to $15$ is given by the cost to prepare $\ket \Phi$, plus the cost to do a round of fault-tolerant error correction. The state $\ket \Phi$ can be prepared by starting from the $\ket 0^8$ state, which is already stabilized by the $Z$-type stabilizers, then fault-tolerantly measuring the $3$ generating $X$-type stabilizers of the Steane code along with the $X^8$ operator. Assuming Shor-style error correction, 
    this requires $d^2 = 9$ rounds, for a total of $36$ measurements, $27$ of which have weight-$4$ and $9$ of which have weight-$8$. Assuming that the measurements are done using Shor-style syndrome extraction as well, this requires producing $27$ four-qubit cat states, as well as $9$ eight-qubit cat states, and using $27\times 4 + 9\times 8 = 180$ CNOTs, measurements, and entangled qubits to extract the syndromes.

    The cost of preparing the cat states can also be estimated, but is usually done via post-selection and hence is a error-rate dependent cost. 
    For the fault-tolerant-to-distance-3 preparation of four- or eight-qubit cat states, however, it is possible to use flag qubits as in Figure~\ref{fig:flagged_cat_states}.
    \begin{figure*}
        \includegraphics[width=0.44\linewidth]{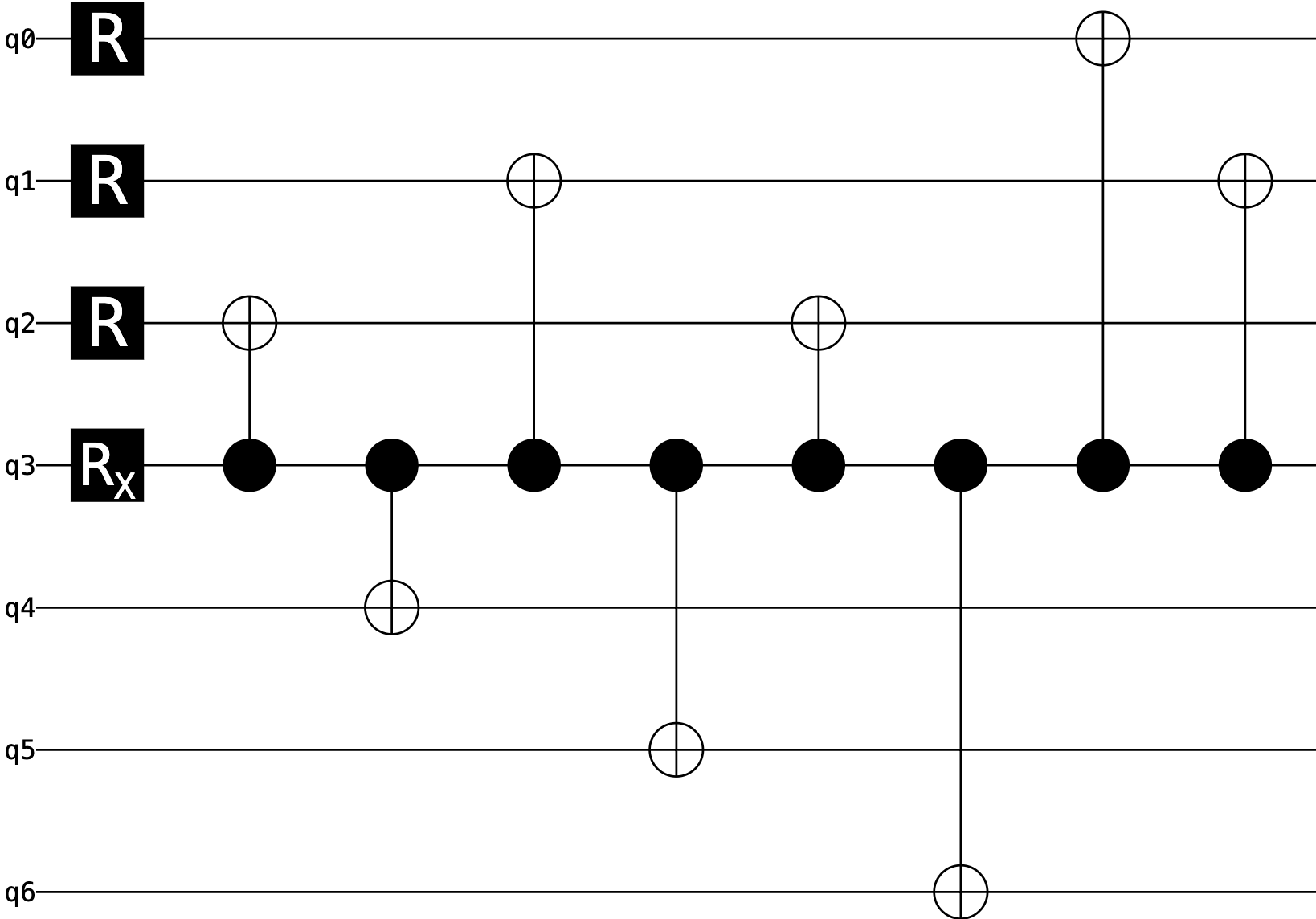}\hfill
        \includegraphics[width=0.54\linewidth]{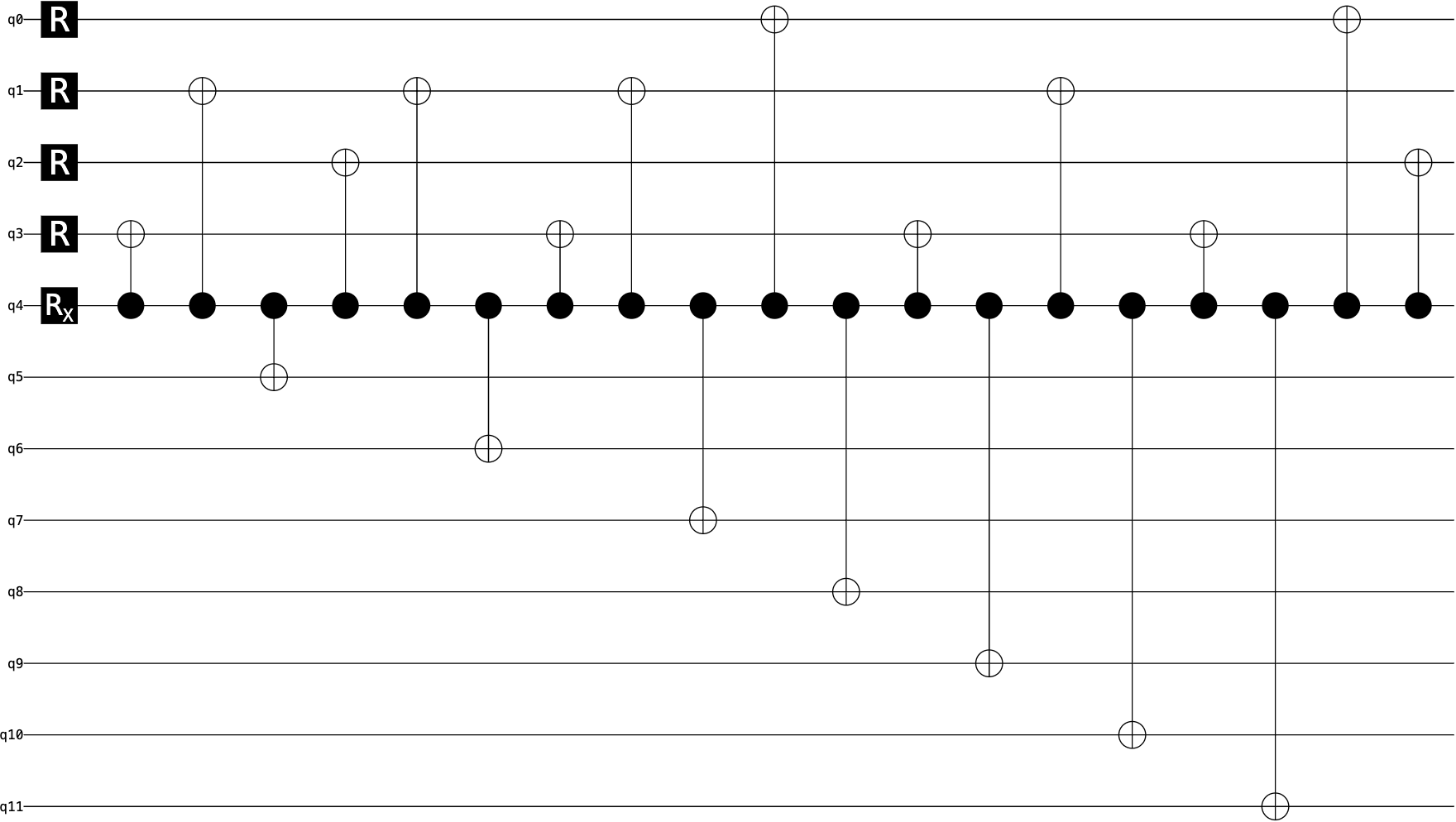}
        \caption{\label{fig:flagged_cat_states}One way to prepare the four- or eight-qubit cat states fault-tolerantly to distance-3. The notation $R$ means prepare in the $\ket 0$ state, $R_X$ prepare in the $\ket +$ state. After measuring the flag qubits prepared in the $Z$ basis, enough information is obtained to correct any data error arising from a weight-one error fault-tolerantly. Flagged state preparation is examined in greater detail in Section~\ref{sec:msp}.}
    \end{figure*}
    This uses three additional qubits and 5 CNOTs  for the preparation of each four-qubit cat state, and 4 additional qubits with $12$ CNOTs for each $8$ qubit state. 
    This adds a total of $27\times 3 + 4\times 9 = 116$ ancilla qubits and $243$ CNOTs.

    Therefore, the difference between doing logical identity, i.e. a fault-tolerant round of error correction, and a logical gate can be reasonably said to be $296$ measurements/non-resettable ancilla qubits and $426$ two-qubit gates using a strategy for universality based on code-switching. In Section~\ref{sec:logical_H_resources} (and Appendix~\ref{sec:random_search}), we will compare the resources required for a flag-qubit based approach for universality.
    
    \subsection{Resources used by our BCH construction}\label{sec:logical_H_resources}
    We now can analyze the resources used by applying the construction outlined in Section~\ref{sec:flagged_cliffords} to the logical $\bar H$ derived above. We use this as an opportunity to precisely specify how to flag an arbitrary circuit.
    Throughout this section, we validate the fault tolerance of each proposed construction numerically.
    This is done by brute-force enumeration of all weight-one errors in the union of the bulk region of the gadgets.
    Notably we do not consider errors that we conceptualize as caught by the error correction or stabilizer measurement  procedure outlined in Figure~\ref{fig:ec_gadget}.

    First, we count the number of CNOTs. In our case, there are $28$. We can then design the first-level flag gadgets. The fact that there are $28$ CNOTs means there are $28$ locations from which $X$ errors can propagate, and $28$ from which $Z$ errors can propagate. 
    In our definition of the main flag gadget in Section~\ref{sec:main_gadget}, each location requires deducing the measurement result of two measurements in the $X$ basis, and two in the $Z$ basis. This means that there are $28\times 2 = 56$ classical bits we will compress for each basis using a parity-check matrix, neglecting the flags protecting the error correction procedure outlined in Section~\ref{sec:boundary_gadget}. The fact that the quantum Reed-Muller code is distance $3$ means that any two errors of weight-$1$ must be distinguishable by their flag patterns. Since any weight-one error on the data can change the measurement result of two of the measurements in the main gadget, this means that the correction power of the code defined by the parity-check matrix defining the flags must be double that of the code, i.e. the distance must be $5$. This is enough to define the first-level flags -- we implement flags according to the parity-check matrix $H$ for the $n = 56, d = 5$ shortened BCH code. To distinguish any possible measurement errors from data errors, our construction prescribes repeating the measurements $3$ times in space, as in Appendix~\ref{sec:random_search}.

    This first level then uses $72$ ancilla qubits and 2566 two qubit gates. 
    We reduce this overhead in two relatively simple ways.
    First, we note that the reason we repeat the flags defined by the parity check matrix several times is to avoid confusing measurement errors, or effective measurement errors, with data errors.
    Since in the unrepeated and uncompressed case any error produces a syndrome of weight at most $2$, if the columns of each parity check matrix are of high enough weight, it seems likely that two repetitions rather than three is sufficient.
    This indeed is validated numerically, which reduces the qubit overhead to $48$ and the gate overhead to $1720$.

    The second method is to optimize the parity-check matrix chosen in order to use as few CNOTs as possible, while still defining the same code, and hence giving us the same correction power and level of fault tolerance. 
    For this example we use a rudimentary approach: just randomly sample equivalent parity-check matrices by taking a random invertible matrix $L$ and a random permutation matrix $R$, and checking the resources used by constructing flag gadgets according to $LHR$ numerically. 
    We optimize the parity-check matrix defining the $X$- and $Z$-type flags separately since they do not interact with one another.
    Using this approach, randomly sampling $1000$ representatives, we find that we can further reduce the number of gates used to $1428$.

    After the first level, we also have to ensure that low-weight (weight-one) errors on the flag qubits do not propagate to high-weight (greater than weight-one) errors on the data. We do this by adding flags to each set of flag gadgets. First we focus on $X$ type flags, i.e. flags which are prepared in the $\ket +$ basis, measured in the $X$ basis, and catch $Z$ errors on the flag qubits. Taking the circuit with $1428$ gates produced by using the parity-check matrix produced by sampling, we follow the same procedure. First, we introduce meta-flags to the $X$-type flags to ensure that $Z$ errors do not propagate to the data. To find what parity-check matrix we need to apply, we count the number of locations from which a $Z$ error can propagate to a nontrivial data error. This is upper bounded by the number of CNOTs connecting to flag qubits prepared in the $X$ basis -- for the sake of convenience we take the number of locations, and hence the block-length of the code defining the second level flags -- equal to this count. Counting the number of locations gives us $618$ for the $X$-type flags and $782$ for the $Z$-type flags (the difference between $618 + 782$ and $1428$ is the number of two qubit gates on the data qubits, which are not touched by the second level flags).

    Similarly to earlier observations about flagged syndrome extraction~\cite{prabhu2023cat}, ensuring that the parity-check matrix has no columns of weight-$1$ should allow us to omit spatial repetition of the parity-check matrices in space. 
    The shortened BCH code on $n = 618$ bits with distance $3$ uses $10$ parity-checks, as does the shortened BCH code on $n = 1428$ bits, meaning we use $20$ additional ancillae for a total of $48 + 20 = 68$. We again perform a rudimentary optimization to achieve a total two-qubit gate count of $8634$ by sampling $200$ random representatives for both $X$ and $Z$ (reduced from 1000 because of the increased time to construct and analyze the larger circuits).
    
    It is notable that in theory one could optimize the entire circuit creation in one loop, using a more sophisticated algorithm than random sampling, making $8634$ a relatively pessimistic upper bound on the number of gates needed for this construction. In summary, though, this method uses approximately $20\times$ as many two-qubit gates and $23\%$ as many ancillae compared to an approach based on code switching (again, not counting the cost of error-correction for either approach, since this is the cost of logical identity which would have been performed either way). Comparison to magic state distillation is more difficult, in that the resources used by state distillation depend on the physical and target error rates. We make a careful examination of the topic in Section~\ref{sec:msp}.

    In Appendix~\ref{sec:random_search} we also consider searching for circuits based on random products of our main gadget to measure space-time stabilizers.
    \section{Application to (Magic) State Preparation}\label{sec:msp}
    Our construction can also be used to fault-tolerantly prepare stabilizer states, or codestates of stabilizer codes. This of course can be viewed as a degenerate case of code-switching between stabilizer codes, which we can also apply our framework to. For the sake of concreteness, we will just provide a few examples of our framework for this application. We briefly touched upon this in Section~\ref{sec:logical_H}, but we now give this subroutine more attention, and apply it to magic state distillation.

    Often the strategy taken for preparing a stabilizer code state is to make projective measurements of each of the stabilizer generators and repeat some number of times to ensure fault tolerance in the presence of measurement errors. Just as in Section~\ref{sec:logical_H}, however, we can instead directly write down a circuit that transforms the stabilizers of our original, physical, state into the stabilizers of the codestate.

    We will in fact use part of the same circuit as in Section~\ref{sec:logical_H}, namely the preparation circuit for the 3D color code. This circuit is shown in Figure~\ref{fig:state_prep}. It is notable that this circuit is actually significantly deeper than the circuit given in Section~\ref{sec:logical_H}, despite being one of the two circuits that was multiplied together to form $\bar H$. 
    This is only slightly surprising: the circuit for logical $H$ only swaps $\bar Z$ and $\bar X$, whereas this circuit must map the single-qubit stabilizers of the unencoded state to the stabilizer generators.
    \begin{figure*}
        \includegraphics[width=0.9\linewidth]{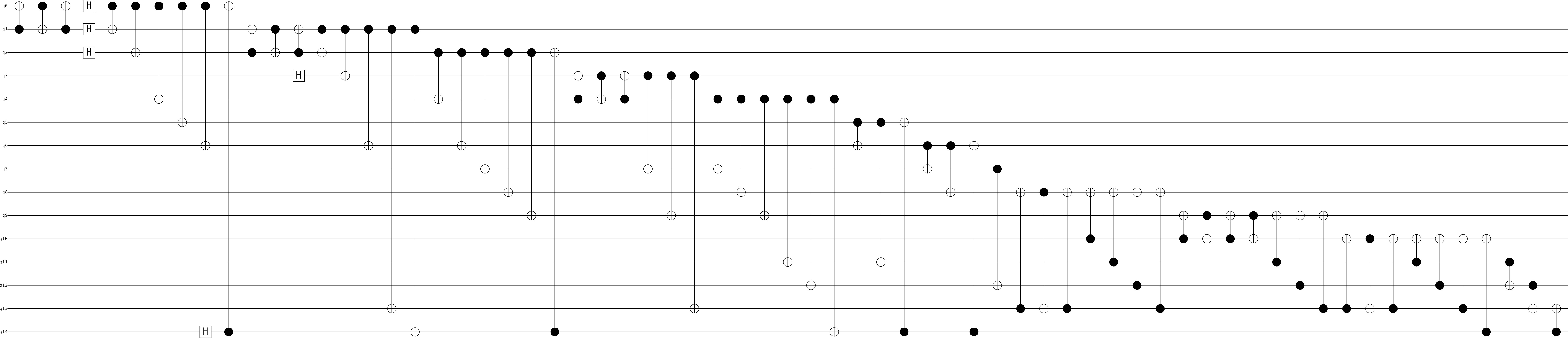}
        \caption{\label{fig:state_prep}One way to prepare the codestate $\ket{\bar \psi}$ in the $[[15, 1, 3]]$ code starting from the state $\ket{0}^{\otimes 14} \otimes \ket{\psi}$.}
    \end{figure*}
    We can now consider flagging the resulting circuit according to Section~\ref{sec:general_flags}. 
    Note that since the initial state is prepared in the trivial code (i.e. not encoded) the initial round of fault-tolerant error correction reduces to ensuring each qubit is prepared in the $\ket{0}$ state, which can be handled by adding one location per qubit to flag to the construction (this one location just being a gadget that measures the $Z$ stabilizer).
    If we flag the resulting circuit without optimizing, we end up with a circuit using 92 ancilla qubits (107 total qubits) and 7380 two-qubit gates (these numbers correspond to a construction in which we eliminate the repetition of the parity-check matrix in the meta-flag level as in Section~\ref{sec:logical_H_resources}, but not in the lower level). 

    The standard way to prepare a stabilizer codestate is to do $O(t^2)$ rounds of projective measurements of a set of stabilizer generators. For concreteness we will take the same stabilizer generators described in Section~\ref{sec:logical_H} which, again, consist of $4$ weight-$8$ $X$-type stabilizer generators and $10$ weight-$4$ $Z$-type stabilizer generators. Following the same resource analysis as we did for code switching in Section~\ref{sec:logical_H} in which we assume reset is extremely slow or unavailable, 
    we estimate the cost of the measurement of a weight-4 operator as 7 ancilla qubits and 9 CNOTs, and the cost of a weight-8 as 12 ancillae and 20 CNOTs. This corresponds to the 4 ancillae and 4 CNOTs (8 ancillae and CNOTs) required to measure the operator using a fault-tolerantly prepared cat state on 4 (8) qubits as well as the 3 ancillae and 5 CNOTs (4 ancillae and 12 CNOTs) we proposed to prepare the cat state itself. Therefore the cost of a single round of error correction is $10 \times 7 + 4 \times 12 = 118$ ancillae and $10 \times 9 + 4 \times 20 = 170$ CNOTs. 

    For the sake of concreteness we assume that exactly $d^2 = 9$ rounds of syndrome extraction are performed in order to ensure fault tolerance. This produces a total cost of $1062$ ancillae and $1530$ CNOTs. Therefore we observe that we can prepare a codestate using less than $10\%$ as many ancillae (or measurements, if reset is fast) while only using about $5\times$ as many CNOTs.

    This directly applies to magic-state distillation, in that one of the most popular magic-state distillation circuits, namely the 15-to-1~\cite{bravyi2005magic} construction, essentially follows this procedure:
    \begin{enumerate}
        \item Prepare a codestate of the $[[15, 1, 3]]$ code
        \item Apply $T$ (using $15$ physical $T^\dagger$ gates applied transversally)
        \item Unencode, producing a single, higher-fidelity, $T$ state
    \end{enumerate}
    Applying our construction, as above, to the encoding circuit means that the space-time volume of magic state distillation could be reduced by roughly a factor of $2$ (although recent work \cite{gidney2024cultivation} has proposed compelling alternatives to distillation for the surface code).
    \section{Flagging Data-Syndrome Codes}\label{sec:data_syn}
    As a final application, we consider making data-syndrome code~\cite{ashikmhin2014robus, fujiwara2014datasyndrome,brown2024datasyndromebch} style syndrome extraction, a natural technique to make syndrome extraction robust to measurement errors, straightforwardly fault tolerant.
    Intuitively, repeated rounds of syndrome extraction constitute encoding the syndrome bits in a repetition code -- by making the distance of the repetition code high enough, that is, performing sufficient rounds of syndrome extraction, one can guarantee that measurement errors are handled fault tolerantly.
    Data-syndrome codes are just a generalization of this idea where we replace the repetition code with a code with some more desirable properties, typically a higher rate.
    The strategy then is to take the parity-check matrix $H$ of some code, produce a redundant parity-check matrix $GH$ using the generator matrix of a classical code so that each check is encoded, measure the new checks, decode according to $G$ to remove measurement errors, then decode this according to $H$ to find the data correction.

    It has been demonstrated~\cite{brown2024datasyndromebch, ashikmhin2014robus, fujiwara2014datasyndrome} that under phenomenological noise, i.e. noise on the data qubits before any operations are performed and before each measurement but not after each physical gate, data-syndrome codes can substantially decrease the logical error rate versus repeated rounds of syndrome extraction. 
    Unfortunately, data-syndrome codes are not fault-tolerant under circuit noise. 
    To be precise, consider an error $e$ on qubit $i$ with syndrome $s$.
     If $e$ occurs before any measurements are performed, then it will produce its syndrome and be corrected. 
     However, if an error on qubit $i$ occurs after the first $x$ measurements, then instead of $s$ a version of $s$ with the first $x$ bits zeroed out will be measured. 
     Decoding this according to the classical code may yield a syndrome corresponding to a data error on a different qubit, making the procedure non-fault-tolerant. 
     In some sense this is because mid-circuit errors act as many measurement errors and the classical code has insufficient distance to correct them.
    Of course, for any particular choice of $G$ it may be the case that a fault-tolerant decoder exists; it simply will probably not be the intuitive two-step decoder outlined above.

    Fortunately, we can apply a flag strategy to make data-syndrome codes fault tolerant under circuit noise. The intuition is this:
    \begin{itemize}
        \item The quantum code itself handles input data errors by construction,
        \item The classical code handles measurement errors,
        \item Mid-circuit errors can be handled by flagging.
    \end{itemize}
    Therefore, to perform error correction fault tolerant to distance $2t + 1$ it is sufficient that that any error $e$ of weight-$|e| \leq t$ which is not equivalent by a (space-time) stabilizer to an error $e'$ with support outside of the region between the first and last CNOT of the syndrome extraction circuit is detected by the flag gadgets (assuming the syndrome extraction circuit itself is fault tolerant as a gadget). For this application, though, the specific form of the syndrome extraction circuit allows us to greatly reduce the complexity of the flag gadgets.

    For the time being, we reduce to CSS codes for the sake of convenience. If we look at a single data qubit, we see one qubit coupled to many - we want to ensure that the state of the qubit is consistent between each of the times it contributes to a parity-check. Effectively, we want to make sure that no (undetected) errors occur between the first and last CNOT connected to the data qubit, marked in Figure~\ref{fig:data_syndrome_bad_region}. To do this we flag the data CNOT. Because of the simple form of the circuit when restricted to this qubit, the flag gadget can just consist of two CNOT, Figure~\ref{fig:data_x_gadget}. If we cover the entirety of the region with $X$ gadgets, we will be sure to detect any errors which produce syndromes the data-syndrome code cannot handle, at least in the absence of flag-qubit measurement errors, Figure~\ref{fig:covered}.
    \begin{figure}
       \centering
        \subfloat[Region in which errors produce inconsistent syndromes.\label{fig:data_syndrome_bad_region}]{%
            \includegraphics[width=0.45\linewidth]{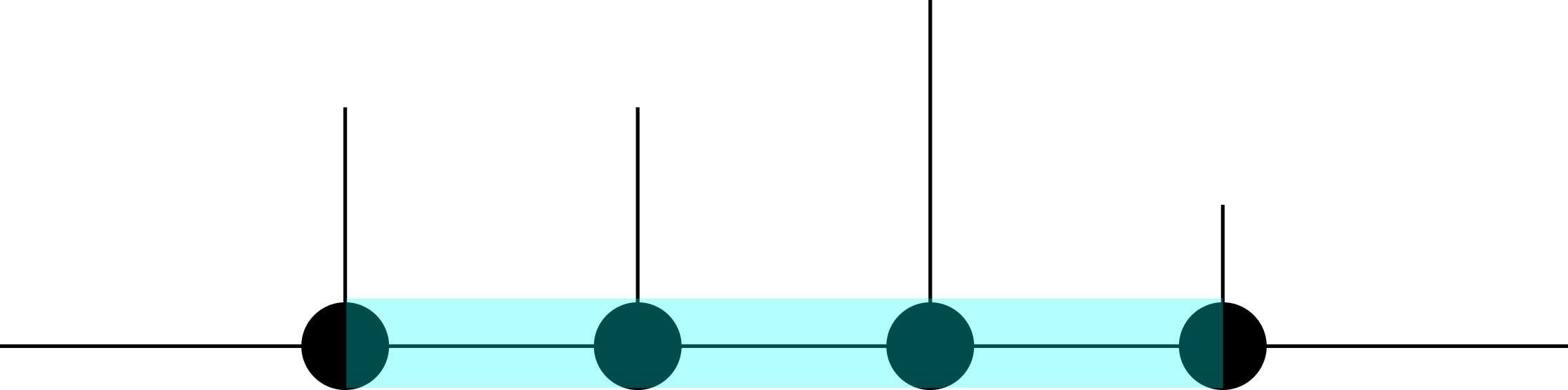}
        }
        \hfill
        \subfloat[The simplified gadget we can use to detect $X$ errors.\label{fig:data_x_gadget}]{%
            \includegraphics[width=0.45\linewidth]{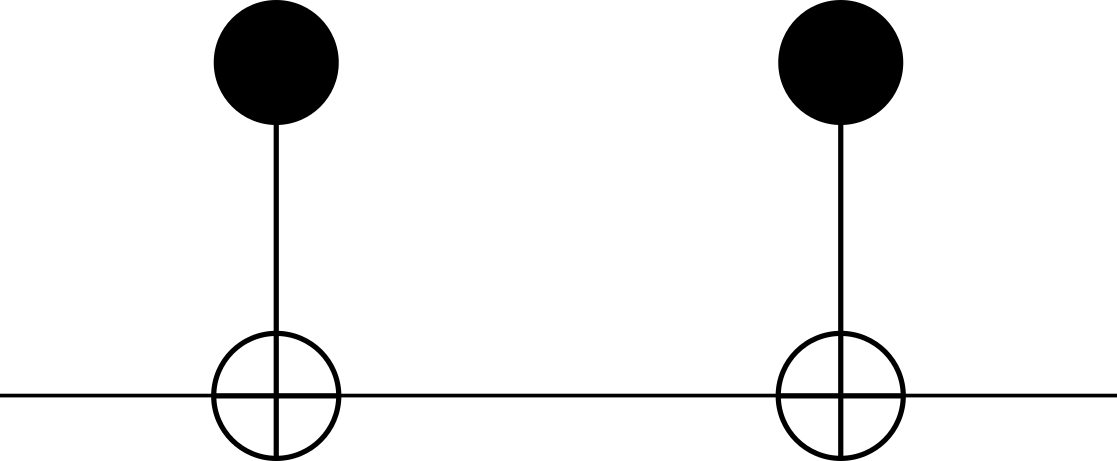}
        }\\
        \subfloat[Covering the entire region with detecting gadgets.\label{fig:covered}]{%
            \includegraphics[width=0.45\linewidth]{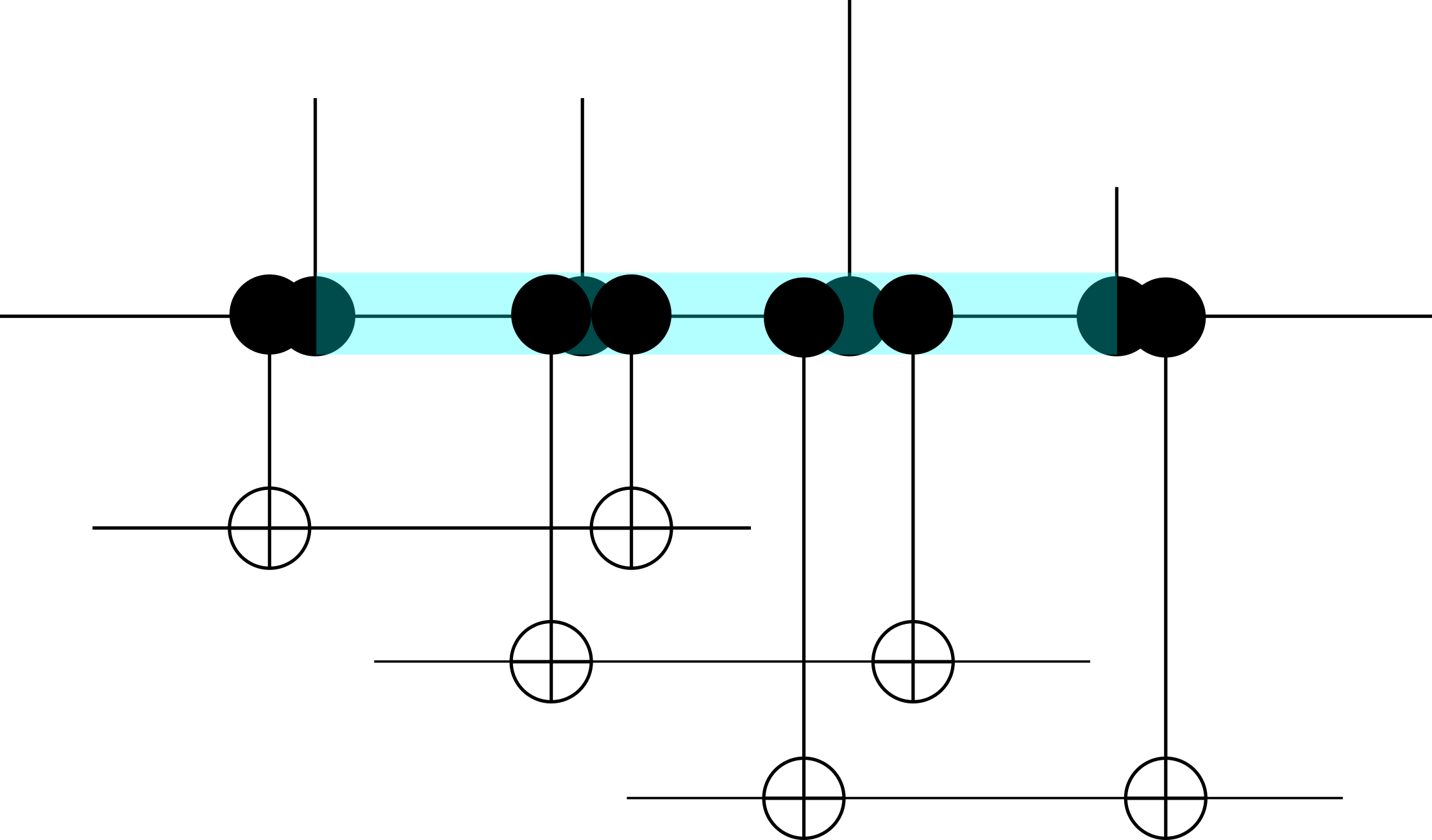}
        }
        \caption{\label{fig:data_flag}Application of flag qubit ideas to data-syndrome codes.} 
    \end{figure}

    Now that we know how to make sure each relevant error is detected, we play the same game as in Section~\ref{sec:general_flags}. We simply add in the locations for each of the other data qubits, combine primitive gadgets according to the parity-checks of some classical code with distance at least $2t + 1$ and repeat the resulting circuit in space $t$ times to account for measurement errors (note that we could recursively apply the same construction to encode the flag measurements in a data-syndrome code). Note that we do not need meta-flags since errors which propagate from flag qubits to data qubits affect only one data qubit, due to the specific form of the syndrome extraction circuit.
    Similarly, although the precision of the resulting gadget is not as high as in our construction for general flag gadgets, this is still fault-tolerant since misidentifying a data error as occurring in a neighboring location only leads to one effective measurement error (of the syndrome) when accounted for.

    It has been shown~\cite{brown2024datasyndromebch} that for an $[[n, k, d]]$ code with $\ell$ independent stabilizer generators, it is possible to find a set of stabilizers to measure that is robust to $t$ measurement errors using only $O(\ell + t \log \ell)$ total measurements. As an upper bound on the number of locations we need to flag, we can assume that each stabilizer has support on $O(n)$ qubits, meaning there are $O(n (\ell + t \log \ell))$ locations to flag. 
    The number of flag qubits needed to make the whole procedure fault-tolerant to distance $2t + 1$ is $O(t^2 \log(n (\ell t \log \ell)) + \ell + t \log \ell)$ (not including the extra measurements, either of a cat state or of more flags to make the original stabilizer measurements fault tolerant).

    Of course, the analysis is nearly exactly the same for non-CSS codes, in that we just have to separate out $X$ and $Z$ locations. We illustrate this using the $5$-qubit perfect code and the parity-check matrix for the Hamming code shortened to the appropriate number of bits for both $X$ and $Z$ flags. 
    We additionally use (a variant of) the original flag qubit prescription for the $[[5, 1, 3]]$~\cite{reichardt2018flag}. The combination of these tools yields a circuit on $28$ qubits using $90$ CNOTs which performs fault-tolerant syndrome extraction for the $[[5, 1, 3]]$ code. A diagram of the relevant pieces as well as the resulting circuit (obtained mostly by hand) is given in Figure~\ref{fig:flagged_perfect}. The most straightforward implementation of a syndrome extraction circuit for the $[[5, 1, 3]]$ code uses four $4$-qubit cat states for each of $4$ rounds of syndrome extraction. Using the same estimates for the cost of producing a $4$ qubit cat state, 3 ancilla qubits and 5 CNOTs, yields a circuit on 112 (non-resettable) qubits using 144 CNOTs.
    
    \begin{figure*}
       \centering
        \subfloat[The stabilizers measurements we make after encoding the syndrome bits of the $5, 1, 3$ code into the $7, 4, 3$ Hamming code.\label{fig:ds}]{%
            \includegraphics[width=0.45\linewidth]{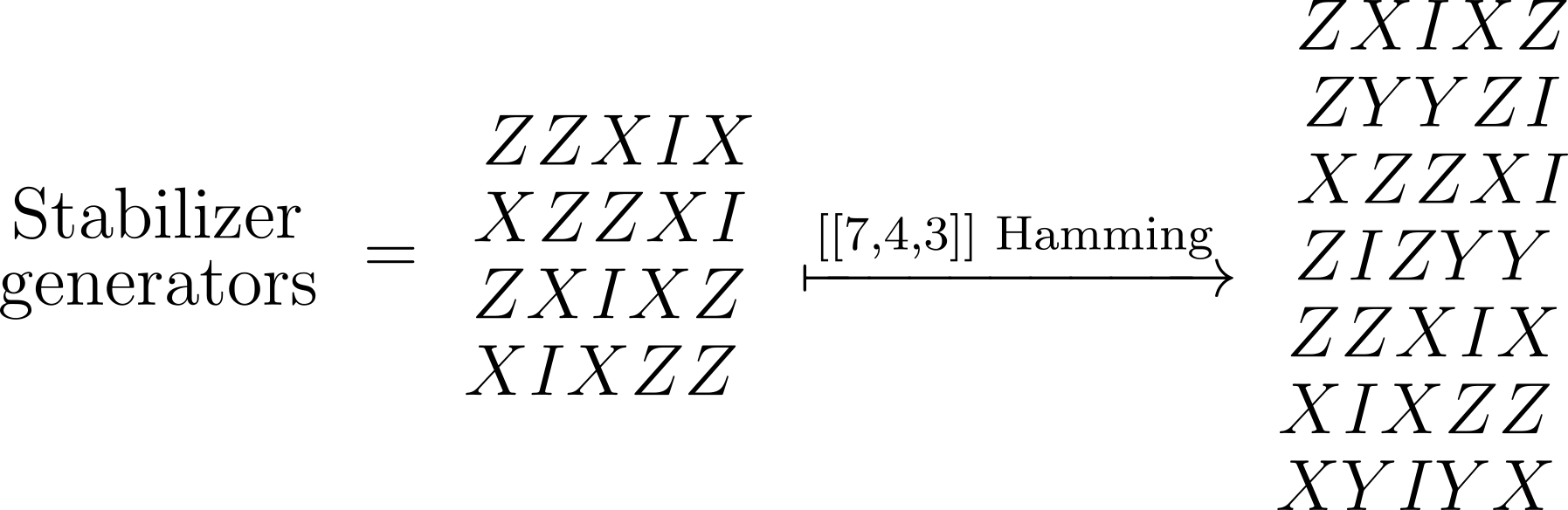}
        }\hfill
        \subfloat[A single stabilizer measurement]{%
            \includegraphics[width=0.15\linewidth]{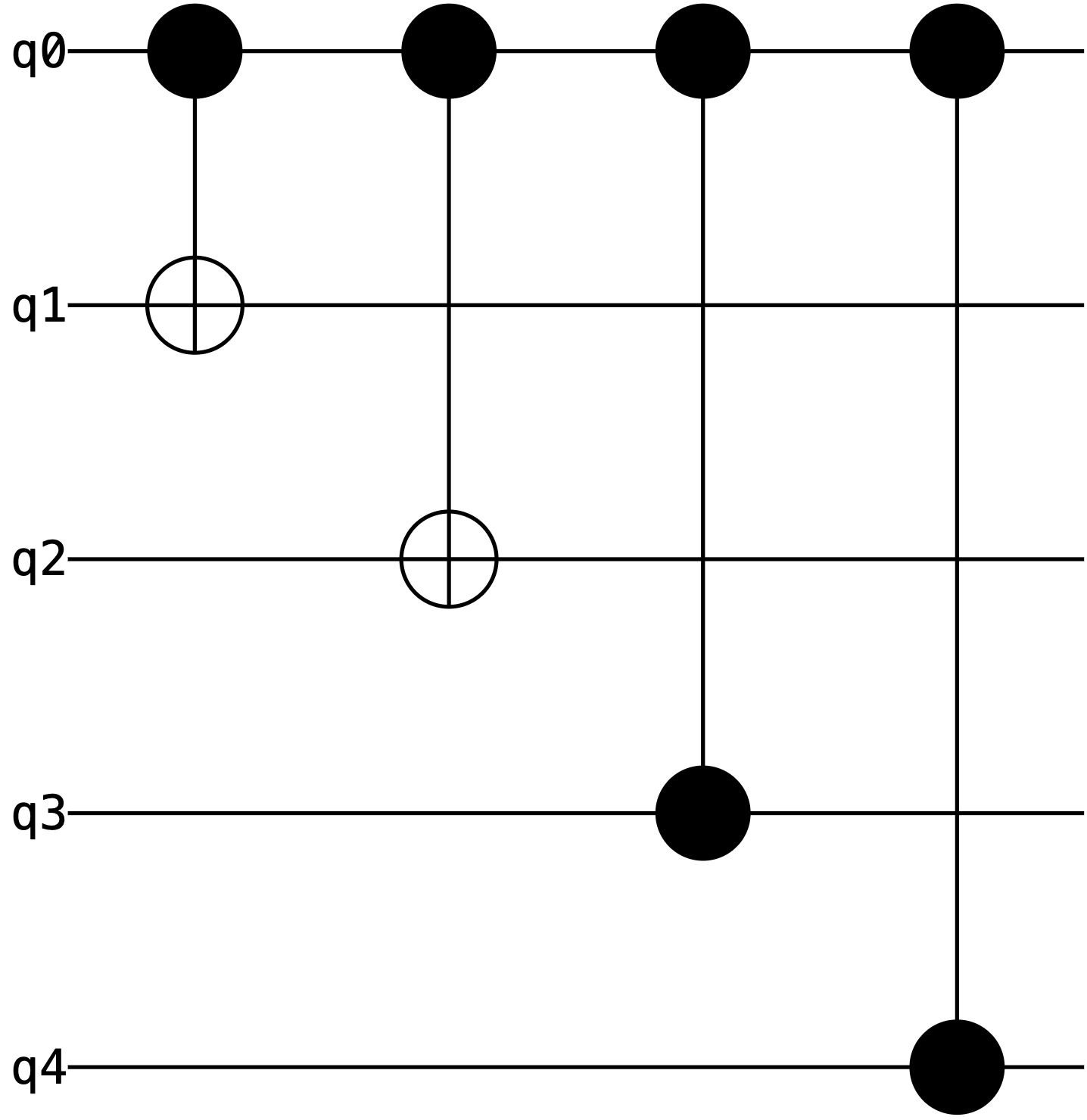}
        }\hfill
        \subfloat[The same stabilizer measurement with a flag due to~\cite{reichardt2018flag}.]{%
            \includegraphics[width=0.2\linewidth]{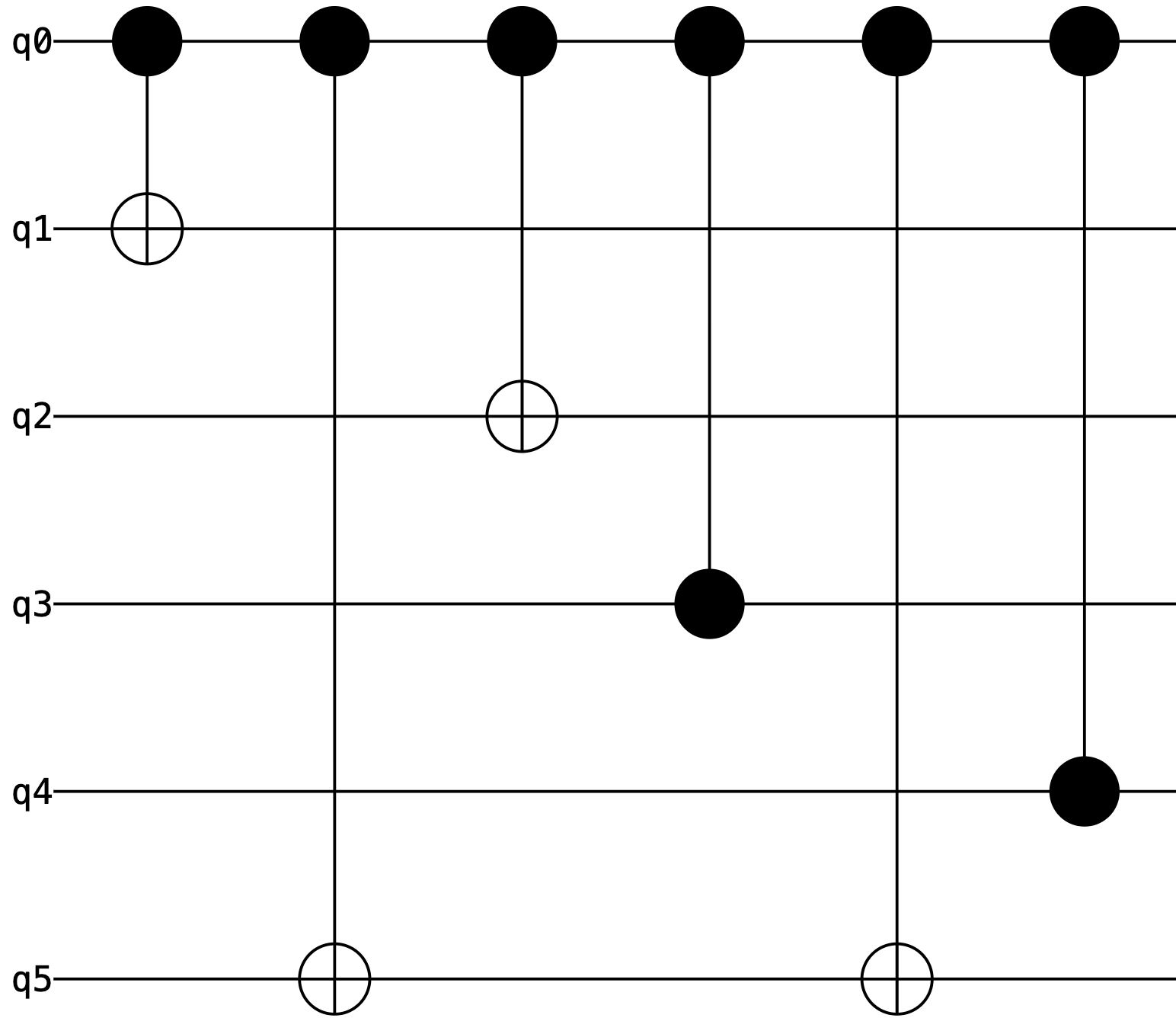}
        }\\
        \subfloat[The resulting circuit, with data qubits highlighted. All other qubits are flag qubits, either for syndrome extraction or for the data. The brackets above the circuit denote CNOTs that (can) occur physically at the same time, but that are drawn separately for clarity.]{%
            \includegraphics[width=\linewidth]{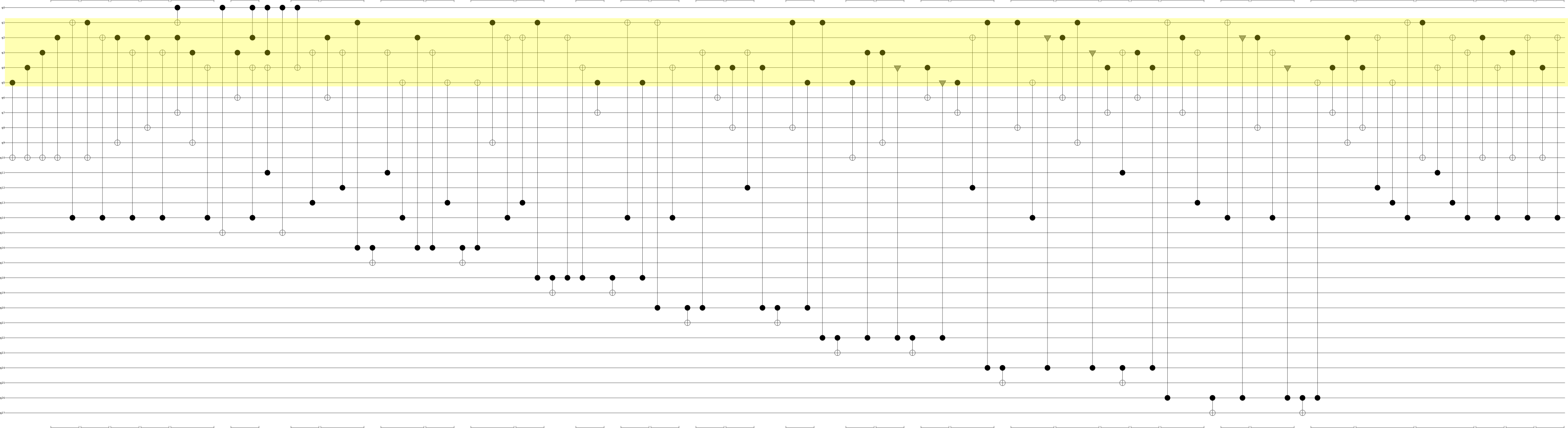}
         }\\ 
        \caption{\label{fig:flagged_perfect}The $5, 1, 3$ perfect code encoded into a flagged data-syndrome code. Crumble circuit available at \href{https://algassert.com/crumble\#circuit=Q(0,0)0;Q(0,1)1;Q(0,2)2;Q(0,3)3;Q(0,4)4;Q(0,5)5;Q(0,6)6;Q(0,7)7;Q(0,8)8;Q(0,9)9;Q(0,10)10;Q(0,11)11;Q(0,12)12;Q(0,13)13;Q(0,14)14;Q(0.5,0.5)15;Q(1,0)16;Q(1.5,0.5)17;Q(2,0)18;Q(2.5,0.5)19;Q(3,0)20;Q(3.5,0.5)21;Q(4,0)22;Q(4.5,0.5)23;Q(5,0)24;Q(5.5,0.5)25;Q(6,0)26;Q(6.5,0.5)27;CX_5_10;TICK;CX_4_10;TICK;CX_3_10;TICK;CX_2_10_14_1;TICK;CX_1_10_14_2;TICK;CX_2_9_14_3;TICK;CX_2_8_14_3;TICK;CX_0_1_2_7_3_9_14_4;TICK;CX_0_15;TICK;CX_3_6_14_4;CZ_0_2;TICK;CX_11_4;CZ_0_3;TICK;CX_0_15;TICK;CX_0_4_13_3;TICK;CX_2_6_12_3;CZ_16_1;TICK;CX_16_17;TICK;CX_11_3_14_5;CZ_16_2;TICK;CX_16_3_13_5;TICK;CX_16_17;TICK;CX_16_5_1_9_14_2;TICK;CX_13_2;CZ_18_1;TICK;CX_18_19;TICK;CX_18_2;TICK;CX_18_4_5_7;TICK;CX_18_19;TICK;CX_14_1;CZ_18_5;TICK;CX_20_1_14_4;TICK;CX_20_21;TICK;CX_20_3_4_6;TICK;CX_4_8_12_3;TICK;CZ_20_4;TICK;CX_20_21;TICK;CX_1_8;CZ_20_5;TICK;CZ_22_1;TICK;CX_22_23;TICK;CX_5_10;CZ_22_3;TICK;CX_3_9;CY_22_4;TICK;CX_22_23;TICK;CX_4_6;CY_22_5;TICK;CX_5_7_12_2;CZ_24_1;TICK;CX_24_25;TICK;CX_1_8_14_5;CY_24_2;TICK;CX_2_6_1_9;CY_24_3;TICK;CX_4_7_11_3_24_25;TICK;CX_3_6;CZ_24_4;TICK;CX_26_1_2_7_13_3;TICK;CX_26_27;TICK;CX_14_1;CY_26_2;TICK;CX_2_8_14_3;CY_26_4;TICK;CX_26_27;TICK;CX_26_5_4_7_2_9;TICK;CX_4_8_12_2_13_5_14_1;TICK;CX_1_10_11_4_13_2_14_3;TICK;CX_2_10_14_4;TICK;CX_3_10_14_2;TICK;CX_4_10_14_2}{this link}.}
    \end{figure*}
    \section{Conclusion}\label{sec:conclusion}
    We have shown that, because of the fact that Paulis can be tracked through Clifford circuits (by definition), these circuits can be made fault tolerant through the use of flag qubits. 
    This observation allows for a new paradigm for fault-tolerant universality, namely transversal non-Clifford gates combined with flagged Clifford gates.
    In many regimes this approach can be less costly in terms of qubits required than conventional methods such as magic-state distillation or code switching. 
    It also allows for a reduction in the resources required for fault-tolerant state preparation and error correction, beyond simply flagging syndrome extraction.

    Even with minimal optimization, we have shown that our construction substantially reduces the cost of fault tolerance across many subroutines. 
    We expect that, given some fraction of the same effort for optimizing older techniques, this method can greatly reduce the resources needed to solve practically useful problems. 

    Beyond general optimization of the specific circuits produced, we also leave open some specific avenues for improvement:
    \begin{itemize}
        \item In which regimes is using a good (classical) LDPC code as the code to combine primitive flag gadgets the most effective choice?
        \item Can this construction be applied in practical regimes for implementing logical Clifford gates on good LDPC codes? One can imagine implementing a (flagged) unencode+operate+re-encode circuit to find fault-tolerant logical operators of a given code, a generally nontrivial task.
        \item For which flag constructions is the resulting decoding problem efficient?
    \end{itemize}

    The question of decoding is probably the most important for practically relevant implementations of this construction. 
    Ultimately, this framework offers a compelling alternative approach to universal fault tolerance, with the potential to reduce resource costs and integrate with existing methods.
    \section{Acknowledgments}
    This work is supported by NSF CAREER award No. CCF-2237356. We would like to thank the UNM Center for Advanced Research Computing, supported in part by the National Science Foundation, for providing the high performance computing resources used in this work. 
    \newpage
    \bibliographystyle{apsrev4-2}
    \bibliography{refs.bib}

\begin{thebibliography}{45}%
\makeatletter
\providecommand \@ifxundefined [1]{%
 \@ifx{#1\undefined}
}%
\providecommand \@ifnum [1]{%
 \ifnum #1\expandafter \@firstoftwo
 \else \expandafter \@secondoftwo
 \fi
}%
\providecommand \@ifx [1]{%
 \ifx #1\expandafter \@firstoftwo
 \else \expandafter \@secondoftwo
 \fi
}%
\providecommand \natexlab [1]{#1}%
\providecommand \enquote  [1]{``#1''}%
\providecommand \bibnamefont  [1]{#1}%
\providecommand \bibfnamefont [1]{#1}%
\providecommand \citenamefont [1]{#1}%
\providecommand \href@noop [0]{\@secondoftwo}%
\providecommand \href [0]{\begingroup \@sanitize@url \@href}%
\providecommand \@href[1]{\@@startlink{#1}\@@href}%
\providecommand \@@href[1]{\endgroup#1\@@endlink}%
\providecommand \@sanitize@url [0]{\catcode `\\12\catcode `\$12\catcode
  `\&12\catcode `\#12\catcode `\^12\catcode `\_12\catcode `\%12\relax}%
\providecommand \@@startlink[1]{}%
\providecommand \@@endlink[0]{}%
\providecommand \url  [0]{\begingroup\@sanitize@url \@url }%
\providecommand \@url [1]{\endgroup\@href {#1}{\urlprefix }}%
\providecommand \urlprefix  [0]{URL }%
\providecommand \Eprint [0]{\href }%
\providecommand \doibase [0]{https://doi.org/}%
\providecommand \selectlanguage [0]{\@gobble}%
\providecommand \bibinfo  [0]{\@secondoftwo}%
\providecommand \bibfield  [0]{\@secondoftwo}%
\providecommand \translation [1]{[#1]}%
\providecommand \BibitemOpen [0]{}%
\providecommand \bibitemStop [0]{}%
\providecommand \bibitemNoStop [0]{.\EOS\space}%
\providecommand \EOS [0]{\spacefactor3000\relax}%
\providecommand \BibitemShut  [1]{\csname bibitem#1\endcsname}%
\let\auto@bib@innerbib\@empty
\bibitem [{\citenamefont {Takagi}\ \emph {et~al.}(2022)\citenamefont {Takagi},
  \citenamefont {Endo}, \citenamefont {Minagawa},\ and\ \citenamefont
  {Gu}}]{takagi2022mitigation}%
  \BibitemOpen
  \bibfield  {author} {\bibinfo {author} {\bibfnamefont {R.}~\bibnamefont
  {Takagi}}, \bibinfo {author} {\bibfnamefont {S.}~\bibnamefont {Endo}},
  \bibinfo {author} {\bibfnamefont {S.}~\bibnamefont {Minagawa}},\ and\
  \bibinfo {author} {\bibfnamefont {M.}~\bibnamefont {Gu}},\ }\bibfield
  {journal} {\bibinfo  {journal} {npj Quantum Information}\ }\textbf {\bibinfo
  {volume} {8}},\ \href {https://doi.org/10.1038/s41534-022-00618-z}
  {10.1038/s41534-022-00618-z} (\bibinfo {year} {2022})\BibitemShut {NoStop}%
\bibitem [{\citenamefont {{Google Quantum
  AI}}(2023)}]{acharya2022suppressingquantumerrorsscaling}%
  \BibitemOpen
  \bibfield  {author} {\bibinfo {author} {\bibnamefont {{Google Quantum AI}}},\
  }\href@noop {} {\bibfield  {journal} {\bibinfo  {journal} {Nature}\ }\textbf
  {\bibinfo {volume} {614}},\ \bibinfo {pages} {676} (\bibinfo {year}
  {2023})}\BibitemShut {NoStop}%
\bibitem [{\citenamefont {Eastin}\ and\ \citenamefont
  {Knill}(2009)}]{eastin2009transversal}%
  \BibitemOpen
  \bibfield  {author} {\bibinfo {author} {\bibfnamefont {B.}~\bibnamefont
  {Eastin}}\ and\ \bibinfo {author} {\bibfnamefont {E.}~\bibnamefont {Knill}},\
  }\href {https://doi.org/10.1103/PhysRevLett.102.110502} {\bibfield  {journal}
  {\bibinfo  {journal} {Physical Review Letters}\ }\textbf {\bibinfo {volume}
  {102}},\ \bibinfo {pages} {110502} (\bibinfo {year} {2009})},\ \bibinfo
  {note} {arXiv:0811.4262 [quant-ph]}\BibitemShut {NoStop}%
\bibitem [{\citenamefont {Bravyi}\ and\ \citenamefont
  {Kitaev}(2005)}]{bravyi2005magic}%
  \BibitemOpen
  \bibfield  {author} {\bibinfo {author} {\bibfnamefont {S.}~\bibnamefont
  {Bravyi}}\ and\ \bibinfo {author} {\bibfnamefont {A.}~\bibnamefont
  {Kitaev}},\ }\href {https://doi.org/10.1103/PhysRevA.71.022316} {\bibfield
  {journal} {\bibinfo  {journal} {Phys. Rev. A}\ }\textbf {\bibinfo {volume}
  {71}},\ \bibinfo {pages} {022316} (\bibinfo {year} {2005})}\BibitemShut
  {NoStop}%
\bibitem [{\citenamefont {Nielsen}\ and\ \citenamefont
  {Chuang}(2012)}]{nielsen2012quantum}%
  \BibitemOpen
  \bibfield  {author} {\bibinfo {author} {\bibfnamefont {M.~A.}\ \bibnamefont
  {Nielsen}}\ and\ \bibinfo {author} {\bibfnamefont {I.~L.}\ \bibnamefont
  {Chuang}},\ }\href@noop {} {\emph {\bibinfo {title} {Quantum Computation and
  Quantum Information}}}\ (\bibinfo  {publisher} {Cambridge University Press},\
  \bibinfo {address} {Cambridge, England},\ \bibinfo {year} {2012})\BibitemShut
  {NoStop}%
\bibitem [{\citenamefont {Paetznick}\ and\ \citenamefont
  {Reichardt}(2013)}]{paetznick2013switch}%
  \BibitemOpen
  \bibfield  {author} {\bibinfo {author} {\bibfnamefont {A.}~\bibnamefont
  {Paetznick}}\ and\ \bibinfo {author} {\bibfnamefont {B.~W.}\ \bibnamefont
  {Reichardt}},\ }\href {https://doi.org/10.1103/PhysRevLett.111.090505}
  {\bibfield  {journal} {\bibinfo  {journal} {Phys. Rev. Lett.}\ }\textbf
  {\bibinfo {volume} {111}},\ \bibinfo {pages} {090505} (\bibinfo {year}
  {2013})}\BibitemShut {NoStop}%
\bibitem [{\citenamefont {Daguerre}\ \emph {et~al.}(2025)\citenamefont
  {Daguerre}, \citenamefont {Blume-Kohout}, \citenamefont {Brown},
  \citenamefont {Hayes},\ and\ \citenamefont {Kim}}]{daguerre2025switch}%
  \BibitemOpen
  \bibfield  {author} {\bibinfo {author} {\bibfnamefont {L.}~\bibnamefont
  {Daguerre}}, \bibinfo {author} {\bibfnamefont {R.}~\bibnamefont
  {Blume-Kohout}}, \bibinfo {author} {\bibfnamefont {N.~C.}\ \bibnamefont
  {Brown}}, \bibinfo {author} {\bibfnamefont {D.}~\bibnamefont {Hayes}},\ and\
  \bibinfo {author} {\bibfnamefont {I.~H.}\ \bibnamefont {Kim}},\ }\href
  {https://arxiv.org/abs/2506.14169} {\bibinfo {title} {Experimental
  demonstration of high-fidelity logical magic states from code switching}}
  (\bibinfo {year} {2025}),\ \Eprint {https://arxiv.org/abs/2506.14169}
  {arXiv:2506.14169 [quant-ph]} \BibitemShut {NoStop}%
\bibitem [{\citenamefont {Butt}\ \emph {et~al.}(2024)\citenamefont {Butt},
  \citenamefont {Heußen}, \citenamefont {Rispler},\ and\ \citenamefont
  {Müller}}]{butt2024switch}%
  \BibitemOpen
  \bibfield  {author} {\bibinfo {author} {\bibfnamefont {F.}~\bibnamefont
  {Butt}}, \bibinfo {author} {\bibfnamefont {S.}~\bibnamefont {Heußen}},
  \bibinfo {author} {\bibfnamefont {M.}~\bibnamefont {Rispler}},\ and\ \bibinfo
  {author} {\bibfnamefont {M.}~\bibnamefont {Müller}},\ }\bibfield  {journal}
  {\bibinfo  {journal} {PRX Quantum}\ }\textbf {\bibinfo {volume} {5}},\ \href
  {https://doi.org/10.1103/prxquantum.5.020345} {10.1103/prxquantum.5.020345}
  (\bibinfo {year} {2024})\BibitemShut {NoStop}%
\bibitem [{\citenamefont {Gidney}\ \emph {et~al.}(2024)\citenamefont {Gidney},
  \citenamefont {Shutty},\ and\ \citenamefont {Jones}}]{gidney2024cultivation}%
  \BibitemOpen
  \bibfield  {author} {\bibinfo {author} {\bibfnamefont {C.}~\bibnamefont
  {Gidney}}, \bibinfo {author} {\bibfnamefont {N.}~\bibnamefont {Shutty}},\
  and\ \bibinfo {author} {\bibfnamefont {C.}~\bibnamefont {Jones}},\ }\href
  {https://arxiv.org/abs/2409.17595} {\bibinfo {title} {Magic state
  cultivation: growing t states as cheap as cnot gates}} (\bibinfo {year}
  {2024}),\ \Eprint {https://arxiv.org/abs/2409.17595} {arXiv:2409.17595
  [quant-ph]} \BibitemShut {NoStop}%
\bibitem [{\citenamefont {Sahay}\ \emph {et~al.}(2025)\citenamefont {Sahay},
  \citenamefont {Tsai}, \citenamefont {Chang}, \citenamefont {Su},
  \citenamefont {Smith}, \citenamefont {Singh},\ and\ \citenamefont
  {Puri}}]{sahay2025cultivation}%
  \BibitemOpen
  \bibfield  {author} {\bibinfo {author} {\bibfnamefont {K.}~\bibnamefont
  {Sahay}}, \bibinfo {author} {\bibfnamefont {P.-K.}\ \bibnamefont {Tsai}},
  \bibinfo {author} {\bibfnamefont {K.}~\bibnamefont {Chang}}, \bibinfo
  {author} {\bibfnamefont {Q.}~\bibnamefont {Su}}, \bibinfo {author}
  {\bibfnamefont {T.~B.}\ \bibnamefont {Smith}}, \bibinfo {author}
  {\bibfnamefont {S.}~\bibnamefont {Singh}},\ and\ \bibinfo {author}
  {\bibfnamefont {S.}~\bibnamefont {Puri}},\ }\href
  {https://arxiv.org/abs/2509.05212} {\bibinfo {title} {Fold-transversal
  surface code cultivation}} (\bibinfo {year} {2025}),\ \Eprint
  {https://arxiv.org/abs/2509.05212} {arXiv:2509.05212 [quant-ph]} \BibitemShut
  {NoStop}%
\bibitem [{\citenamefont {Chao}\ and\ \citenamefont
  {Reichardt}(2018)}]{reichardt2018flag}%
  \BibitemOpen
  \bibfield  {author} {\bibinfo {author} {\bibfnamefont {R.}~\bibnamefont
  {Chao}}\ and\ \bibinfo {author} {\bibfnamefont {B.~W.}\ \bibnamefont
  {Reichardt}},\ }\href {https://doi.org/10.1103/PhysRevLett.121.050502}
  {\bibfield  {journal} {\bibinfo  {journal} {Physical Review Letters}\
  }\textbf {\bibinfo {volume} {121}},\ \bibinfo {pages} {050502} (\bibinfo
  {year} {2018})},\ \bibinfo {note} {arXiv:1705.02329 [quant-ph]}\BibitemShut
  {NoStop}%
\bibitem [{\citenamefont {Gottesman}(1998)}]{gottesman1998simulation}%
  \BibitemOpen
  \bibfield  {author} {\bibinfo {author} {\bibfnamefont {D.}~\bibnamefont
  {Gottesman}}\ }\href {https://doi.org/10.48550/arXiv.quant-ph/9807006}
  {10.48550/arXiv.quant-ph/9807006} (\bibinfo {year} {1998}),\ \bibinfo {note}
  {arXiv:quant-ph/9807006}\BibitemShut {NoStop}%
\bibitem [{\citenamefont {Prabhu}\ and\ \citenamefont
  {Reichardt}(2023)}]{prabhu2023cat}%
  \BibitemOpen
  \bibfield  {author} {\bibinfo {author} {\bibfnamefont {P.}~\bibnamefont
  {Prabhu}}\ and\ \bibinfo {author} {\bibfnamefont {B.~W.}\ \bibnamefont
  {Reichardt}},\ }\href {https://doi.org/10.22331/q-2023-10-24-1154} {\bibfield
   {journal} {\bibinfo  {journal} {{Quantum}}\ }\textbf {\bibinfo {volume}
  {7}},\ \bibinfo {pages} {1154} (\bibinfo {year} {2023})}\BibitemShut
  {NoStop}%
\bibitem [{\citenamefont {Forlivesi}\ and\ \citenamefont
  {Amaro}(2025)}]{forlivesi2025flag}%
  \BibitemOpen
  \bibfield  {author} {\bibinfo {author} {\bibfnamefont {D.}~\bibnamefont
  {Forlivesi}}\ and\ \bibinfo {author} {\bibfnamefont {D.}~\bibnamefont
  {Amaro}},\ }\href {https://arxiv.org/abs/2508.14200} {\bibinfo {title} {Flag
  at origin: a modular fault-tolerant preparation for css codes}} (\bibinfo
  {year} {2025}),\ \Eprint {https://arxiv.org/abs/2508.14200} {arXiv:2508.14200
  [quant-ph]} \BibitemShut {NoStop}%
\bibitem [{\citenamefont {Chamberland}\ \emph {et~al.}(2020)\citenamefont
  {Chamberland}, \citenamefont {Kubica}, \citenamefont {Yoder},\ and\
  \citenamefont {Zhu}}]{chamberland2020triangular}%
  \BibitemOpen
  \bibfield  {author} {\bibinfo {author} {\bibfnamefont {C.}~\bibnamefont
  {Chamberland}}, \bibinfo {author} {\bibfnamefont {A.}~\bibnamefont {Kubica}},
  \bibinfo {author} {\bibfnamefont {T.~J.}\ \bibnamefont {Yoder}},\ and\
  \bibinfo {author} {\bibfnamefont {G.}~\bibnamefont {Zhu}},\ }\href@noop {}
  {\bibfield  {journal} {\bibinfo  {journal} {New Journal of Physics}\ }\textbf
  {\bibinfo {volume} {22}},\ \bibinfo {pages} {023019} (\bibinfo {year}
  {2020})}\BibitemShut {NoStop}%
\bibitem [{\citenamefont {Tansuwannont}\ \emph {et~al.}(2020)\citenamefont
  {Tansuwannont}, \citenamefont {Chamberland},\ and\ \citenamefont
  {Leung}}]{tansuwannont2020flag}%
  \BibitemOpen
  \bibfield  {author} {\bibinfo {author} {\bibfnamefont {T.}~\bibnamefont
  {Tansuwannont}}, \bibinfo {author} {\bibfnamefont {C.}~\bibnamefont
  {Chamberland}},\ and\ \bibinfo {author} {\bibfnamefont {D.}~\bibnamefont
  {Leung}},\ }\href {https://doi.org/10.1103/PhysRevA.101.012342} {\bibfield
  {journal} {\bibinfo  {journal} {Phys. Rev. A}\ }\textbf {\bibinfo {volume}
  {101}},\ \bibinfo {pages} {012342} (\bibinfo {year} {2020})}\BibitemShut
  {NoStop}%
\bibitem [{\citenamefont {Debroy}\ and\ \citenamefont
  {Brown}(2020)}]{debroy2020cpc}%
  \BibitemOpen
  \bibfield  {author} {\bibinfo {author} {\bibfnamefont {D.~M.}\ \bibnamefont
  {Debroy}}\ and\ \bibinfo {author} {\bibfnamefont {K.~R.}\ \bibnamefont
  {Brown}},\ }\href {https://doi.org/10.1103/PhysRevA.102.052409} {\bibfield
  {journal} {\bibinfo  {journal} {Physical Review A}\ }\textbf {\bibinfo
  {volume} {102}},\ \bibinfo {pages} {052409} (\bibinfo {year} {2020})},\
  \bibinfo {note} {arXiv:2009.07752 [quant-ph]}\BibitemShut {NoStop}%
\bibitem [{\citenamefont {Delfosse}\ and\ \citenamefont
  {Tham}(2025)}]{delfosse2025clifford}%
  \BibitemOpen
  \bibfield  {author} {\bibinfo {author} {\bibfnamefont {N.}~\bibnamefont
  {Delfosse}}\ and\ \bibinfo {author} {\bibfnamefont {E.}~\bibnamefont
  {Tham}},\ }\href {https://doi.org/10.1103/PhysRevLett.134.090603} {\bibfield
  {journal} {\bibinfo  {journal} {Physical Review Letters}\ }\textbf {\bibinfo
  {volume} {134}},\ \bibinfo {pages} {090603} (\bibinfo {year}
  {2025})}\BibitemShut {NoStop}%
\bibitem [{\citenamefont {Tham}\ and\ \citenamefont
  {Delfosse}(2025)}]{Tham2025optimizedclifford}%
  \BibitemOpen
  \bibfield  {author} {\bibinfo {author} {\bibfnamefont {E.}~\bibnamefont
  {Tham}}\ and\ \bibinfo {author} {\bibfnamefont {N.}~\bibnamefont
  {Delfosse}},\ }\href {https://doi.org/10.22331/q-2025-08-27-1829} {\bibfield
  {journal} {\bibinfo  {journal} {{Quantum}}\ }\textbf {\bibinfo {volume}
  {9}},\ \bibinfo {pages} {1829} (\bibinfo {year} {2025})}\BibitemShut
  {NoStop}%
\bibitem [{\citenamefont {Martiel}\ and\ \citenamefont
  {Javadi-Abhari}(2025)}]{martiel2025lowoverheaderrordetectionspacetime}%
  \BibitemOpen
  \bibfield  {author} {\bibinfo {author} {\bibfnamefont {S.}~\bibnamefont
  {Martiel}}\ and\ \bibinfo {author} {\bibfnamefont {A.}~\bibnamefont
  {Javadi-Abhari}},\ }\href {https://arxiv.org/abs/2504.15725} {\bibinfo
  {title} {Low-overhead error detection with spacetime codes}} (\bibinfo {year}
  {2025}),\ \Eprint {https://arxiv.org/abs/2504.15725} {arXiv:2504.15725
  [quant-ph]} \BibitemShut {NoStop}%
\bibitem [{\citenamefont {Delfosse}\ and\ \citenamefont
  {Paetznick}(2023)}]{delfosse2023spacetime}%
  \BibitemOpen
  \bibfield  {author} {\bibinfo {author} {\bibfnamefont {N.}~\bibnamefont
  {Delfosse}}\ and\ \bibinfo {author} {\bibfnamefont {A.}~\bibnamefont
  {Paetznick}},\ }\href {https://arxiv.org/abs/2304.05943} {\bibinfo {title}
  {Spacetime codes of clifford circuits}} (\bibinfo {year} {2023}),\ \Eprint
  {https://arxiv.org/abs/2304.05943} {arXiv:2304.05943 [quant-ph]} \BibitemShut
  {NoStop}%
\bibitem [{\citenamefont {Vasmer}\ and\ \citenamefont
  {Browne}(2019)}]{vasmer20193d}%
  \BibitemOpen
  \bibfield  {author} {\bibinfo {author} {\bibfnamefont {M.}~\bibnamefont
  {Vasmer}}\ and\ \bibinfo {author} {\bibfnamefont {D.~E.}\ \bibnamefont
  {Browne}},\ }\href {https://doi.org/10.1103/PhysRevA.100.012312} {\bibfield
  {journal} {\bibinfo  {journal} {Physical Review A}\ }\textbf {\bibinfo
  {volume} {100}},\ \bibinfo {pages} {012312} (\bibinfo {year} {2019})},\
  \bibinfo {note} {arXiv:1801.04255 [quant-ph]}\BibitemShut {NoStop}%
\bibitem [{\citenamefont {Knill}\ \emph {et~al.}(1996)\citenamefont {Knill},
  \citenamefont {Laflamme},\ and\ \citenamefont {Zurek}}]{knill1996reedmuller}%
  \BibitemOpen
  \bibfield  {author} {\bibinfo {author} {\bibfnamefont {E.}~\bibnamefont
  {Knill}}, \bibinfo {author} {\bibfnamefont {R.}~\bibnamefont {Laflamme}},\
  and\ \bibinfo {author} {\bibfnamefont {W.}~\bibnamefont {Zurek}}\ }\href
  {https://doi.org/10.48550/arXiv.quant-ph/9610011}
  {10.48550/arXiv.quant-ph/9610011} (\bibinfo {year} {1996}),\ \bibinfo {note}
  {arXiv:quant-ph/9610011}\BibitemShut {NoStop}%
\bibitem [{\citenamefont {Koutsioumpas}\ \emph {et~al.}(2022)\citenamefont
  {Koutsioumpas}, \citenamefont {Banfield},\ and\ \citenamefont
  {Kay}}]{koutsioumpas2022smallestT}%
  \BibitemOpen
  \bibfield  {author} {\bibinfo {author} {\bibfnamefont {S.}~\bibnamefont
  {Koutsioumpas}}, \bibinfo {author} {\bibfnamefont {D.}~\bibnamefont
  {Banfield}},\ and\ \bibinfo {author} {\bibfnamefont {A.}~\bibnamefont {Kay}}\
  }\href {https://doi.org/10.48550/arXiv.2210.14066}
  {10.48550/arXiv.2210.14066} (\bibinfo {year} {2022}),\ \bibinfo {note}
  {arXiv:2210.14066 [quant-ph]}\BibitemShut {NoStop}%
\bibitem [{\citenamefont {Fujiwara}(2014)}]{fujiwara2014datasyndrome}%
  \BibitemOpen
  \bibfield  {author} {\bibinfo {author} {\bibfnamefont {Y.}~\bibnamefont
  {Fujiwara}},\ }\href {https://doi.org/10.1103/PhysRevA.90.062304} {\bibfield
  {journal} {\bibinfo  {journal} {Physical Review A}\ }\textbf {\bibinfo
  {volume} {90}},\ \bibinfo {pages} {062304} (\bibinfo {year}
  {2014})}\BibitemShut {NoStop}%
\bibitem [{\citenamefont {Ashikhmin}\ \emph {et~al.}(2014)\citenamefont
  {Ashikhmin}, \citenamefont {Lai},\ and\ \citenamefont
  {Brun}}]{ashikmhin2014robus}%
  \BibitemOpen
  \bibfield  {author} {\bibinfo {author} {\bibfnamefont {A.}~\bibnamefont
  {Ashikhmin}}, \bibinfo {author} {\bibfnamefont {C.-Y.}\ \bibnamefont {Lai}},\
  and\ \bibinfo {author} {\bibfnamefont {T.~A.}\ \bibnamefont {Brun}},\ }in\
  \href {https://doi.org/10.1109/ISIT.2014.6874892} {\emph {\bibinfo
  {booktitle} {2014 IEEE International Symposium on Information Theory}}}\
  (\bibinfo {year} {2014})\ pp.\ \bibinfo {pages} {546--550}\BibitemShut
  {NoStop}%
\bibitem [{\citenamefont {Guttentag}\ \emph {et~al.}(2023)\citenamefont
  {Guttentag}, \citenamefont {Nemec},\ and\ \citenamefont
  {Brown}}]{brown2024datasyndromebch}%
  \BibitemOpen
  \bibfield  {author} {\bibinfo {author} {\bibfnamefont {E.}~\bibnamefont
  {Guttentag}}, \bibinfo {author} {\bibfnamefont {A.}~\bibnamefont {Nemec}},\
  and\ \bibinfo {author} {\bibfnamefont {K.~R.}\ \bibnamefont {Brown}}\ }\href
  {https://doi.org/10.48550/arXiv.2311.16044} {10.48550/arXiv.2311.16044}
  (\bibinfo {year} {2023}),\ \bibinfo {note} {arXiv:2311.16044
  [quant-ph]}\BibitemShut {NoStop}%
\bibitem [{\citenamefont {Calderbank}\ and\ \citenamefont
  {Shor}(1996)}]{calderbank1996css}%
  \BibitemOpen
  \bibfield  {author} {\bibinfo {author} {\bibfnamefont {A.~R.}\ \bibnamefont
  {Calderbank}}\ and\ \bibinfo {author} {\bibfnamefont {P.~W.}\ \bibnamefont
  {Shor}},\ }\href {https://doi.org/10.1103/physreva.54.1098} {\bibfield
  {journal} {\bibinfo  {journal} {Physical Review A}\ }\textbf {\bibinfo
  {volume} {54}},\ \bibinfo {pages} {1098–1105} (\bibinfo {year}
  {1996})}\BibitemShut {NoStop}%
\bibitem [{\citenamefont {Steane}()}]{steane1996css}%
  \BibitemOpen
  \bibfield  {author} {\bibinfo {author} {\bibfnamefont {A.}~\bibnamefont
  {Steane}},\ }\bibfield  {journal} {\bibinfo  {journal} {Proceedings of the
  Royal Society of London. Series A: Mathematical, Physical and Engineering
  Sciences}\ }\textbf {\bibinfo {volume} {452}},\ \href
  {https://doi.org/10.1098/rspa.1996.0136} {10.1098/rspa.1996.0136}\BibitemShut
  {NoStop}%
\bibitem [{\citenamefont {Campbell}(2019)}]{campbell2019singleshot}%
  \BibitemOpen
  \bibfield  {author} {\bibinfo {author} {\bibfnamefont {E.~T.}\ \bibnamefont
  {Campbell}},\ }\href {https://doi.org/10.1088/2058-9565/aafc8f} {\bibfield
  {journal} {\bibinfo  {journal} {Quantum Science and Technology}\ }\textbf
  {\bibinfo {volume} {4}},\ \bibinfo {pages} {025006} (\bibinfo {year}
  {2019})},\ \bibinfo {note} {arXiv:1805.09271 [quant-ph]}\BibitemShut
  {NoStop}%
\bibitem [{\citenamefont {Knill}\ \emph {et~al.}(1998)\citenamefont {Knill},
  \citenamefont {Laflamme},\ and\ \citenamefont {Zurek}}]{knill1998threshold}%
  \BibitemOpen
  \bibfield  {author} {\bibinfo {author} {\bibfnamefont {E.}~\bibnamefont
  {Knill}}, \bibinfo {author} {\bibfnamefont {R.}~\bibnamefont {Laflamme}},\
  and\ \bibinfo {author} {\bibfnamefont {W.~H.}\ \bibnamefont {Zurek}},\ }\href
  {https://doi.org/10.1126/science.279.5349.342} {\bibfield  {journal}
  {\bibinfo  {journal} {Science}\ }\textbf {\bibinfo {volume} {279}},\ \bibinfo
  {pages} {342–345} (\bibinfo {year} {1998})}\BibitemShut {NoStop}%
\bibitem [{\citenamefont {Aharonov}\ and\ \citenamefont
  {Ben-Or}(2008)}]{aharonov2008threshold}%
  \BibitemOpen
  \bibfield  {author} {\bibinfo {author} {\bibfnamefont {D.}~\bibnamefont
  {Aharonov}}\ and\ \bibinfo {author} {\bibfnamefont {M.}~\bibnamefont
  {Ben-Or}},\ }\href {https://doi.org/10.1137/S0097539799359385} {\bibfield
  {journal} {\bibinfo  {journal} {SIAM Journal on Computing}\ }\textbf
  {\bibinfo {volume} {38}},\ \bibinfo {pages} {1207} (\bibinfo {year}
  {2008})},\ \Eprint
  {https://arxiv.org/abs/https://doi.org/10.1137/S0097539799359385}
  {https://doi.org/10.1137/S0097539799359385} \BibitemShut {NoStop}%
\bibitem [{\citenamefont {Campbell}\ \emph {et~al.}(2017)\citenamefont
  {Campbell}, \citenamefont {Terhal},\ and\ \citenamefont
  {Vuillot}}]{campbell2017roads}%
  \BibitemOpen
  \bibfield  {author} {\bibinfo {author} {\bibfnamefont {E.~T.}\ \bibnamefont
  {Campbell}}, \bibinfo {author} {\bibfnamefont {B.~M.}\ \bibnamefont
  {Terhal}},\ and\ \bibinfo {author} {\bibfnamefont {C.}~\bibnamefont
  {Vuillot}},\ }\href {https://doi.org/10.1038/nature23460} {\bibfield
  {journal} {\bibinfo  {journal} {Nature}\ }\textbf {\bibinfo {volume} {549}},\
  \bibinfo {pages} {172–179} (\bibinfo {year} {2017})}\BibitemShut {NoStop}%
\bibitem [{\citenamefont {Gidney}(2021)}]{gidney2021stim}%
  \BibitemOpen
  \bibfield  {author} {\bibinfo {author} {\bibfnamefont {C.}~\bibnamefont
  {Gidney}},\ }\href {https://doi.org/10.22331/q-2021-07-06-497} {\bibfield
  {journal} {\bibinfo  {journal} {{Quantum}}\ }\textbf {\bibinfo {volume}
  {5}},\ \bibinfo {pages} {497} (\bibinfo {year} {2021})}\BibitemShut {NoStop}%
\bibitem [{\citenamefont {Anker}\ and\ \citenamefont
  {Marvian}(2025)}]{anker2025compressed}%
  \BibitemOpen
  \bibfield  {author} {\bibinfo {author} {\bibfnamefont {B.}~\bibnamefont
  {Anker}}\ and\ \bibinfo {author} {\bibfnamefont {M.}~\bibnamefont
  {Marvian}},\ }\href {https://arxiv.org/abs/2509.07288} {\bibinfo {title}
  {Compressing syndrome measurement sequences}} (\bibinfo {year} {2025}),\
  \Eprint {https://arxiv.org/abs/2509.07288} {arXiv:2509.07288 [quant-ph]}
  \BibitemShut {NoStop}%
\bibitem [{\citenamefont {Gottesman}(2022)}]{gottesman2022perspective}%
  \BibitemOpen
  \bibfield  {author} {\bibinfo {author} {\bibfnamefont {D.}~\bibnamefont
  {Gottesman}},\ }\href {https://doi.org/10.48550/arXiv.2210.15844} {\bibinfo
  {title} {Opportunities and challenges in fault-tolerant quantum computation}}
  (\bibinfo {year} {2022}),\ \bibinfo {note} {arXiv:2210.15844
  [quant-ph]}\BibitemShut {NoStop}%
\bibitem [{\citenamefont {Hocquenghem}(1959)}]{hocquenghem1959bch}%
  \BibitemOpen
  \bibfield  {author} {\bibinfo {author} {\bibfnamefont {A.}~\bibnamefont
  {Hocquenghem}},\ }\href@noop {} {\bibfield  {journal} {\bibinfo  {journal}
  {Chiffres (Paris)}\ ,\ \bibinfo {pages} {147}} (\bibinfo {year}
  {1959})}\BibitemShut {NoStop}%
\bibitem [{\citenamefont {Bose}\ and\ \citenamefont {K.}(1960)}]{bose1960bch}%
  \BibitemOpen
  \bibfield  {author} {\bibinfo {author} {\bibfnamefont {R.~C.}\ \bibnamefont
  {Bose}}\ and\ \bibinfo {author} {\bibfnamefont {R.-C.~D.}\ \bibnamefont
  {K.}},\ }\href
  {https://doi.org/https://doi.org/10.1016/S0019-9958(60)90287-4} {\bibfield
  {journal} {\bibinfo  {journal} {Information and Control}\ }\textbf {\bibinfo
  {volume} {3}},\ \bibinfo {pages} {68} (\bibinfo {year} {1960})}\BibitemShut
  {NoStop}%
\bibitem [{\citenamefont {Anker}\ and\ \citenamefont
  {Marvian}(2024)}]{anker2024flag}%
  \BibitemOpen
  \bibfield  {author} {\bibinfo {author} {\bibfnamefont {B.}~\bibnamefont
  {Anker}}\ and\ \bibinfo {author} {\bibfnamefont {M.}~\bibnamefont
  {Marvian}},\ }\href {https://doi.org/10.1103/PRXQuantum.5.040340} {\bibfield
  {journal} {\bibinfo  {journal} {PRX Quantum}\ }\textbf {\bibinfo {volume}
  {5}},\ \bibinfo {pages} {040340} (\bibinfo {year} {2024})}\BibitemShut
  {NoStop}%
\bibitem [{\citenamefont {Kissinger}\ and\ \citenamefont {van~de
  Wetering}(2020)}]{kissinger2020pyzx}%
  \BibitemOpen
  \bibfield  {author} {\bibinfo {author} {\bibfnamefont {A.}~\bibnamefont
  {Kissinger}}\ and\ \bibinfo {author} {\bibfnamefont {J.}~\bibnamefont {van~de
  Wetering}},\ }in\ \href {https://doi.org/10.4204/EPTCS.318.14} {\emph
  {\bibinfo {booktitle} {{\rm Proceedings 16th International Conference on}
  Quantum Physics and Logic, {\rm Chapman University, Orange, CA, USA., 10-14
  June 2019}}}},\ \bibinfo {series} {Electronic Proceedings in Theoretical
  Computer Science}, Vol.\ \bibinfo {volume} {318},\ \bibinfo {editor} {edited
  by\ \bibinfo {editor} {\bibfnamefont {B.}~\bibnamefont {Coecke}}\ and\
  \bibinfo {editor} {\bibfnamefont {M.}~\bibnamefont {Leifer}}}\ (\bibinfo
  {publisher} {Open Publishing Association},\ \bibinfo {year} {2020})\ pp.\
  \bibinfo {pages} {229--241}\BibitemShut {NoStop}%
\bibitem [{\citenamefont {Steane}(1996)}]{steane1996steanecode}%
  \BibitemOpen
  \bibfield  {author} {\bibinfo {author} {\bibfnamefont {A.}~\bibnamefont
  {Steane}},\ }\href {https://doi.org/10.1098/rspa.1996.0136} {\bibfield
  {journal} {\bibinfo  {journal} {Proceedings of the Royal Society of London.
  Series A: Mathematical, Physical and Engineering Sciences}\ }\textbf
  {\bibinfo {volume} {452}},\ \bibinfo {pages} {2551–2577} (\bibinfo {year}
  {1996})},\ \bibinfo {note} {arXiv:quant-ph/9601029}\BibitemShut {NoStop}%
\bibitem [{\citenamefont {Kubica}\ and\ \citenamefont
  {Beverland}(2015)}]{kubica2015color}%
  \BibitemOpen
  \bibfield  {author} {\bibinfo {author} {\bibfnamefont {A.}~\bibnamefont
  {Kubica}}\ and\ \bibinfo {author} {\bibfnamefont {M.~E.}\ \bibnamefont
  {Beverland}},\ }\href {https://doi.org/10.1103/PhysRevA.91.032330} {\bibfield
   {journal} {\bibinfo  {journal} {Physical Review A}\ }\textbf {\bibinfo
  {volume} {91}},\ \bibinfo {pages} {032330} (\bibinfo {year}
  {2015})}\BibitemShut {NoStop}%
\bibitem [{\citenamefont {Bombín}(2015)}]{bombin2015color}%
  \BibitemOpen
  \bibfield  {author} {\bibinfo {author} {\bibfnamefont {H.}~\bibnamefont
  {Bombín}},\ }\href {https://doi.org/10.1088/1367-2630/17/8/083002}
  {\bibfield  {journal} {\bibinfo  {journal} {New Journal of Physics}\ }\textbf
  {\bibinfo {volume} {17}},\ \bibinfo {pages} {083002} (\bibinfo {year}
  {2015})}\BibitemShut {NoStop}%
\bibitem [{\citenamefont {Anderson}\ \emph {et~al.}(2014)\citenamefont
  {Anderson}, \citenamefont {Duclos-Cianci},\ and\ \citenamefont
  {Poulin}}]{poulin2014color}%
  \BibitemOpen
  \bibfield  {author} {\bibinfo {author} {\bibfnamefont {J.~T.}\ \bibnamefont
  {Anderson}}, \bibinfo {author} {\bibfnamefont {G.}~\bibnamefont
  {Duclos-Cianci}},\ and\ \bibinfo {author} {\bibfnamefont {D.}~\bibnamefont
  {Poulin}},\ }\href {https://doi.org/10.1103/PhysRevLett.113.080501}
  {\bibfield  {journal} {\bibinfo  {journal} {Physical Review Letters}\
  }\textbf {\bibinfo {volume} {113}},\ \bibinfo {pages} {080501} (\bibinfo
  {year} {2014})},\ \bibinfo {note} {arXiv:1403.2734 [quant-ph]}\BibitemShut
  {NoStop}%
\bibitem [{\citenamefont {Heu{\ss{}}en}\ and\ \citenamefont
  {Hilder}(2025)}]{heussen2025efficientfault}%
  \BibitemOpen
  \bibfield  {author} {\bibinfo {author} {\bibfnamefont {S.}~\bibnamefont
  {Heu{\ss{}}en}}\ and\ \bibinfo {author} {\bibfnamefont {J.}~\bibnamefont
  {Hilder}},\ }\href {https://doi.org/10.22331/q-2025-09-03-1846} {\bibfield
  {journal} {\bibinfo  {journal} {{Quantum}}\ }\textbf {\bibinfo {volume}
  {9}},\ \bibinfo {pages} {1846} (\bibinfo {year} {2025})}\BibitemShut
  {NoStop}%
\end{thebibliography}%
    \clearpage
    \appendix
    \subsection{Numerical Search over Random Matrices}\label{sec:random_search}
    In the main text we focused on one implementation of a general provably fault-tolerant construction for any Clifford circuit. 
    However, the resources used by this construction are an upper bound on the resources used by a flag gadget construction tailored to the Clifford circuit in question. To avoid measuring logical information, we will still need to restrict ourselves to measuring (products of) space-time stabilizers. These are just products of the two gadgets measured to create the gadget in Section~\ref{sec:main_gadget}. Therefore we can still restrict ourselves to creating gadgets based upon a binary matrix.
    
    Here, we will do a random numerical search for appropriate flag gadgets. That is, we will look for four binary matrices, $A, A', B, B'$ which define, respectively, $X$-type flags, their meta-flags, $Z$-type flags, and their meta-flags such that any nontrivial error $e$ (including measurement errors) that is undetectable either
    \begin{enumerate}
        \item has weight at least $3$
        \item propagates to an error of weight at most $|e|$ on the data.
    \end{enumerate}
    This ensures that there exists a fault-tolerant decoding scheme based upon the measurement results, while allowing more freedom than necessarily distinguishing all weight-$1$ errors from each other.

    Although randomly sampling matrices using significantly fewer resources than the version based upon the BCH code did not yield a provably fault-tolerant circuit, we did find several examples such that the proportion of randomly sampled weight-1 or -2 errors which both propagated to a higher weight error and gave a trivial flag pattern was less than $1\%$, using fewer than $400$ additional gates and $24$ additional qubits. The unflagged version of $\bar H$ we have presented, by contrast, takes $81\%$ of weight-1 or -2 errors to a higher weight error, without any possibility for error detection.

    This suggests that, at least for sufficiently low error rates, minimum-weight corrections based upon the flag data can allow the randomly flagged version of $\bar H$ to produce a lower logical error rate than the unflagged version, despite the fact that it is not completely fault tolerant. Indeed, this is supported in numerical simulations, summarized in Figure~\ref{fig:non_FT_error_reduction}.
    \begin{figure}
        \includegraphics[width=0.9\columnwidth]{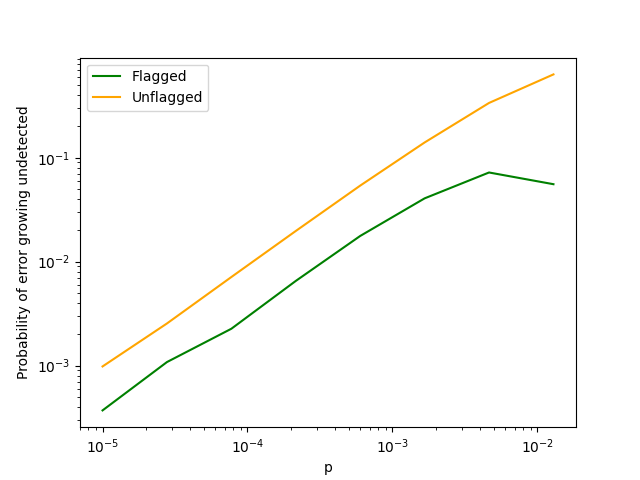}
        \caption{\label{fig:non_FT_error_reduction}Here we plot, in green, the probability of an error of weight-$x$ propagating to an error of weight at least $x + 1$ without triggering any flag qubits. We plot the probability of an error of weight-$x$ propagating to an error of weight at least $x + 1$ on the unflagged circuit in orange. We can see that the probability of a low-weight error propagating to a high-weight error, roughly, the probability of violating fault tolerance and a proxy for logical-error-rate scaling, is reduced by adding random flags. Each data point was obtained by taking $100,000$ random samples. The error model we used introduces an $X$ or $Z$ error before each gate with equal probability $p$ (and hence $Y$ errors with probability $p^2$).}
    \end{figure}
    This is significant because the matrices we sampled were extremely sparse - a density (probability of a check bit being $1$) of $0.05$ for $A, B$ and $0.1$ for $A', B'$. Taking these as parity-check matrices for a code yields codes with trivial distances almost always. Nevertheless, even adding such a small amount of structure can reduce the error rate.

    We do not consider applying our construction using a good LDPC code for much the same reasons we have kept the density of the matrices considered in this section so low -- a random sparse matrix is, at least heuristically, a parity-check matrix for an LDPC code with good distance. The fact that, in this regime, random matrices do not manage to use fewer resources than the BCH construction while maintaining the distance suggests the same for good LDPC codes. For circuits larger than the relatively small $28$ two-qubit gate implementation of $\bar H$ we have taken as our example, an approach based on good LDPC codes would probably be advantageous.
\end{document}